\documentclass[journal,transmag]{IEEEtran}

\usepackage{algorithm}
\usepackage{algorithmicx}
\usepackage{algpseudocode}
\usepackage{amssymb}
\usepackage{amsmath}
\usepackage{amsfonts}
\usepackage{arydshln} 
\usepackage{bm}
\usepackage{color}
\usepackage{graphicx}
\usepackage{multirow} 
\usepackage{soul}
\usepackage{amssymb}
\usepackage{subfigure}
\usepackage{float}
\usepackage{mathtools}
\usepackage[utf8]{inputenc}
\usepackage[english]{babel}
\usepackage{amsthm}
\usepackage[nopostdot,nogroupskip,nonumberlist]{glossaries}
\usepackage{cite}
\usepackage{lipsum}
\usepackage{clipboard}
\newclipboard{clipboard_tsp2020}

\newacronym{ADAS}{ADAS}{advanced driver-assistance system}
\newacronym{TBP}{TBP}{time-bandwidth product}
\newacronym{CWS}{CWS}{collision warning system}
\newacronym{TTC}{TTC}{time-to-collision}
\newacronym{SNR}{SNR}{signal-to-noise ratio}
\newacronym{MLE}{MLE}{maximum likelihood estimation}
\newacronym{TWDL}{TWDL}{total wrong decision loss}
\newacronym{PWDL}{PWDL}{point wrong decision loss}
\newacronym{CRLB}{CRLB}{Cramer-Rao lower bound}
\newacronym{FMCW}{FMCW}{frequency modulated continuous wave}
\newacronym{MF}{MF}{matched filtering}
\newacronym{GLRT}{GLRT}{generalized likelihood ratio test}
\newacronym{FIM}{FIM}{Fisher information matrix}
\newacronym{MTWDL}{MTWDL}{minimal total wrong decision loss}
\newacronym{OFDM}{OFDM}{orthogonal frequency division multiplexing}
\newacronym{MSE}{MSE}{mean square error}

\usepackage[switch]{lineno}
\usepackage[bookmarks,colorlinks,bookmarksnumbered]{hyperref}
\hypersetup{colorlinks=true,pdfauthor=author}

\DeclarePairedDelimiter\abs{\lvert}{\rvert}%
\DeclarePairedDelimiter\norm{\lVert}{\rVert}%

\makeatletter
\let\oldabs\abs
\def\abs{\@ifstar{\oldabs}{\oldabs*}}
\let\oldnorm\norm
\def\norm{\@ifstar{\oldnorm}{\oldnorm*}}
\makeatother

\makeatletter
\renewcommand{\fnum@figure}{Fig. \thefigure}
\makeatother

\newcommand{\diag}{\mbox{diag}}

\newtheorem{theorem}{Theorem}

\newtheorem{prop}{Proposition}

\graphicspath{{figures/}}

\begin{document}
	\title{Designing the Waveform Bandwidth and Time Duration of Automotive Radars for Better Collision Warning Performance}
	
	\author{
		Hang~Ruan, 
		Yimin~Liu,~\IEEEmembership{Member,~IEEE,}
		Tianyao~Huang,
		and~Xiqin~Wang
		\thanks{
			H. Ruan, Y. Liu, T. Huang and X. Wang are with the Department of Electronic
			Engineering, Tsinghua University, Beijing, 100084 China (email:
			ruanh15@mails.tsinghua.edu.cn; yiminliu@tsinghua.edu.cn; huangtianyao@tsinghua.edu.cn; wangxq\_ee@tsinghua.edu.cn).
			Y. Liu is the corresponding author.
		}
		\thanks{
			This work was supported by the National Natural Science Foundation of China (Grant No. 61801258).
		}
		\thanks{
			Some preliminary results of this work were submitted the IEEE Radar Conference 2020.
		}
	}
	
	\maketitle
	
	\begin{abstract}
		Automotive radar is a key component in an \gls{ADAS}
		. The increasing number of radars implemented in vehicles makes interference between them a noteworthy issue. One method of interference mitigation is to limit the \gls{TBP} of radar waveforms. However, the problems of how much \gls{TBP} is necessary and how to optimally utilize the limited \gls{TBP} have not been addressed. We take \gls{CWS} as an example and propose a method of designing the radar waveform parameters oriented by the  performance of \gls{CWS}. We propose a metric to quantify the \gls{CWS} performance and study how the radar waveform parameters (bandwidth and duration) influence this metric. Then, the waveform parameters are designed with a limit on the \gls{TBP} to optimize the system performance. Numerical results show that the proposed design outperforms the state-of-the-art parameter settings in terms of system performance and resource or energy efficiency.
	\end{abstract}
	\vspace{-0.3cm}
	
	
	\section{Introduction}
	\label{sec:intro}
	%
	%
	%
	%
	
	An \gls{ADAS} is designed to promote quality of driving experience, including safety, comfort and fuel economy \cite{shaout2011advanced}. There are different types of \gls{ADAS} with varying purposes. Collision warning/collision avoidance (CW/CA) systems aim to detect potential collisions with objects in front of the vehicle and take corresponding actions, and adaptive cruise control (ACC) systems focus on speed and distance control in the car-following process \cite{vahidi2003research}.
	For such systems to work, a vehicle is equipped with sensors to obtain information from its surrounding. Automotive radar plays an indispensable role among these sensors due to its high measurement accuracy, long detection range and robustness in different weather situations \cite{mukhtar2015vehicle}. 
	
	
	
	A radar measures the range and velocity of object in the front, which are two essential inputs for an \gls{ADAS}. In a \gls{CWS}, indices, such as \gls{TTC} and warning distance, are calculated with these inputs and the system decides whether to give a warning accordingly \cite{lee2016real}. In an ACC system, a control is generated based on the measurements
	\cite{li2011model}.
	
	The range and velocity are measured by estimating the delay and Doppler shift of the received signal \cite{skolnik1962introduction}. Besides high \gls{SNR}, there are requirements on the bandwidth and duration of radar waveform to obtain accurate estimates. The delay can be obtained through pulse compression and high accuracy of delay requires great bandwidth of the waveform. The Doppler shift can be obtained through coherent processing of a pulse sequence and high accuracy of velocity requires long duration of the waveform \cite{skolnik1962introduction}.
	
	To ensure the performance of an \gls{ADAS}, one essential radar requirement is to obtain accurate measurements. As introduced above, range measurement with high resolution and accuracy demands a wide spectrum, while velocity measurement demands a long duration \cite{rohling2001waveform}. In other words, a radar requires adequate electromagnetic resources in both the frequency and time domains to obtain reliable measurements. With the popularization of automotive radars, however, it will be common that multiple radars operate at the same time and in the same neighborhood \cite{kunert2012eu}. The wider spectrum and longer duration of a single radar waveforms will result in a greater possibility of spectral conflict between radars, resulting in issues such as an increased noise floor 
	and ghost targets \cite{Brooker2007Mutual,goppelt2010automotive}. These problems degrade radar's performance in estimation and detection and this further leads to performance degradation of the \gls{ADAS}. Therefore, countermeasures have to be implemented to mitigate interference, which can be classified into the following two schemes shown in Fig.~\ref{fig:schemes}:
	
	\begin{enumerate}
		\item Overlapped sharing scheme: The entire electromagnetic resources, in the frequency and time domains, are divided into several subsets. Each subset contains part of the whole bandwidth and time slots with no mutual overlap. Each radar uses some of the subsets according to a certain strategy. In this case, spectral conflict is likely to occur due to lack of coordination. Metrics such as outage probability are studied to evaluate the scheme \cite{braun2013co-channel, alhourani2018stochastic}. \Copy{polarization}{The interference can be further mitigated by implementing interference cancellation, beamforming and polarization \cite{Brooker2007Mutual, fischer2011minimizing, kunert2012eu}. These techniques, however, increase the signal processing complexity, increase the hardware cost \cite{sun2015interference} and achieve only limited SINR improvement \cite{kunert2012eu}.}
		\item Non-overlapped allocation scheme: \Copy{radarmac}{Centralized control is assumed to be feasible in this scheme (such control may be realized through vehicle-to-vehicle (V2V) and vehicle-to-infrastructure (V2I) communications \cite{papadimitratos2009vehicular}), and the entire resources (in the frequency and time domains) are allocated to radars in a conflict-free way according to individual demands \cite{ruan2016sharing, khoury2016radarmac,aydogdu2019radar,aydogdu2019radarchat}. Ideally, spectral interference can be suppressed or even avoided. For example, \cite{khoury2016radarmac} proposes a mechanism called RadarMAC, in which a control center collects the positions and trajectories of vehicles in some region first and then allocate the resources to vehicles without overlap.} 
	\end{enumerate}
	
	\begin{figure}[t]
		\begin{minipage}[b]{1.0\linewidth}
			\centering
			\centerline{\includegraphics[width=8.8cm]{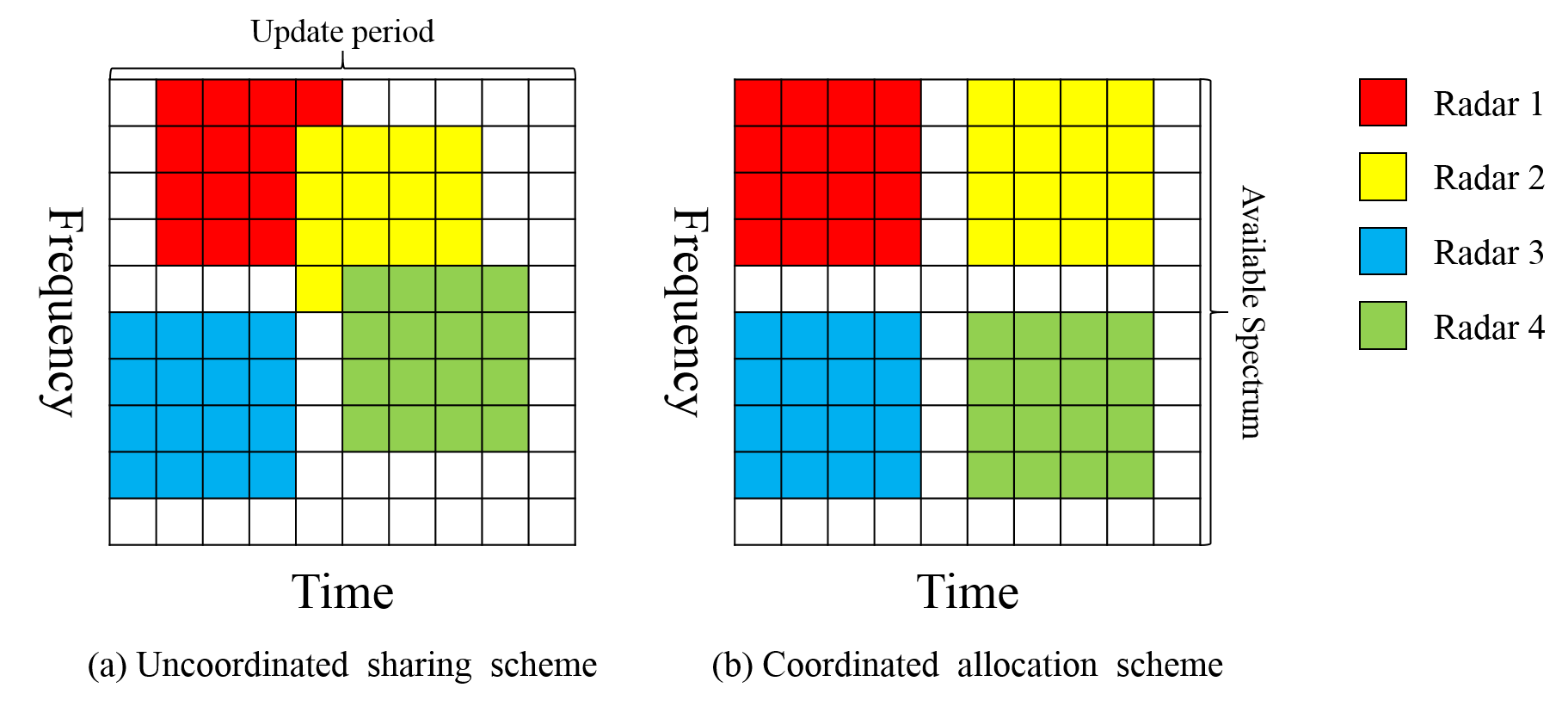}}
		\end{minipage}
		\caption{A sketch of (a) the uncoordinated sharing scheme and (b) the coordinated allocation scheme. In the coordinated scheme, guard intervals are reserved in both the time and frequency domains to avoid signal overlap due to time delay and Doppler shift, respectively.}
		\label{fig:schemes}
		\vspace{-0.4cm}
	\end{figure}
	
	\Copy{TBP_quantify}{For both sharing and allocation schemes, the \gls{TBP} of the radar waveform is a key parameter since it quantifies how many resources in the frequency and time domains are occupied by the radar. For the sharing scheme, greater \gls{TBP} occupied by one radar leads to a greater probability of spectral conflict \cite{braun2013co-channel,alhourani2018stochastic}, while for allocation, greater \gls{TBP} occupied leads to fewer radars that can operate in the same place and at the same time without interference \cite{khoury2016radarmac,aydogdu2019radar}. Therefore, the interference problem can be mitigated if the \gls{TBP} of radar waveform is as small as possible, or in other words, the radar occupies as few electromagnetic resources as possible. This result motivates us to take resource efficiency into consideration when designing radar waveform.}
	
	Generally speaking, the existing studies on radar waveform design can be classified into waveform parameter design and baseband waveform design. As for waveform parameter design, parameters are designed to meet certain requirements on radar performance. One typical example is that requirements on measuring (e.g., range and velocity resolution) are proposed and the corresponding parameters (e.g., bandwidth and duration) are designed accordingly \cite{russell1997millimeter,hasch2012millimeter,sit2018bpsk,yi201924}. As for baseband waveform design, the \gls{FMCW} is widely used in due to its low transmit power \cite{folster2005automotive,forstner200877ghz}, while the \gls{OFDM} is another alternative due to its flexibility in waveform optimization and convenience in integration with communication systems \cite{braun2009parametrization,turlapaty2014range}. \Copy{existing_optimization}{The baseband waveform can be further optimized to improve radar's performance. Metrics, including mutual information \cite{bell1993information,leshem2007information} and
		\gls{CRLB} \cite{lenz2018joint,turlapaty2014range} are used for evaluating radar's performance in the optimization problems. The ambiguity function is also used for optimization since it affects radar's range and velocity resolution \cite{knapp1976gneralized,chen2008mimo}.
	
	Moreover, when the electromagnetic resources occupied by a radar, quantified by \gls{TBP}, are limited, there is a tradeoff of estimation accuracy between range and velocity. Authors in \cite{stein2014information,lenz2018joint} express the tradeoff through multiplying the \gls{CRLB} of range and velocity by a weighting matrix. The authors in \cite{lenz2018joint} further explore the Pareto-optimal region of delay-Doppler estimation by changing the weighting matrix.} \Copy{compare_lenz}{However, the tradeoff in \cite{lenz2018joint} is made with a limited on the transmit power instead of TBP. Moreover, it is not addressed how to design the weighting matrix and find an optimal tradeoff. In this respect, the relationship between radar performance and system performance expresses one way of tradeoff. The main purpose of our paper is to express the tradeoff in terms of system and to design the waveform by finding the optimum of tradeoff.}

	\Copy{conventional_our}{To summarize, many existing studies on waveform design are targeted to specific radar performances, meaning that they consider the radar as a separated sensor from the \gls{ADAS}. They do not jointly consider how the radar performance influences the system performance. Therefore, the electromagnetic resources are not fully exploited in terms of system performance. When the issue of interference among automotive radars becomes increasingly severe, however, it is necessary for the radars to utilize the limited  resources more efficiently. Understanding the relationship between radar performance and system performance is a key to improving the efficiency.
	
	\Copy{outline}{Based on the above discussion, we propose a novel method of designing radar waveform parameters oriented to the performance of \gls{CWS}. 
		We evaluate the system performance by introducing the \gls{MTWDL} as the metric, which quantifies the costs of the system when a warning is triggered wrongly in a safe case or is not triggered timely in a dangerous situation. We then analyze the relationship between this metric and radar parameters (bandwidth and duration), based on which waveform parameters are designed through minimizing the \gls{MTWDL} of the system under the \gls{TBP} constraint.
		Compared with the existing design method, our method enables the radar to use fewer electromagnetic resources while guaranteeing the \gls{CWS} to function normally.}}
	\Copy{contribution}{The main contributions of this paper are as follows:

	\begin{enumerate}
		\item We introduce \gls{MTWDL} as the performance metric of \gls{CWS} and analyze its relationship with radar performance. Analysis reveals that the system performance relies on what we defined as range-velocity joint error index (called error index for short), which is a weighted sum of variances of range and velocity estimation. This result indicates that in order to improve the performance of \gls{CWS}, we should minimize the error index, yielding the waveform design criterion. 
		\item Using the above criterion, we  propose a method optimizing the waveform parameters 
		under the constraint of \gls{TBP}, which limits the resources consumed by radar. Both theoretical and numerical results show that our optimized parameter design outperforms the state-of-art-design in terms of system performance, resource efficiency and energy efficiency.
		\item To quantify the relationship between system performance and waveform parameters, we evaluate the estimation error of radar. To this end, we derive asymptotic the \gls{CRLB} of range and velocity estimation using the \gls{FMCW} in a multi-target scenario. 
	\end{enumerate}}
	
	In summary, the proposed method raises waveform design to the level of the system that radar serves. By mining the relationship between system performance and waveform parameters, this method enables the radar to utilize the limited electromagnetic resources more efficiently and has the potential to mitigate the tense resources. Our methodology can be applied to waveform design of radar in other systems, especially those also facing the problem of resource shortage.  
	
	The remainder of this paper is arranged as follows: In Section \ref{sec:model}, we formulate the signal model, derive \gls{CRLB} of range and velocity and study the distribution of estimation errors.
	In Section \ref{sec:performance}, we specify the collision warning system in this study, including its hypothesis, decision rule and derive the performance metric. In Section \ref{sec:design}, we propose a method of waveform parameter design to optimize the system performance. The numerical results are shown in Section \ref{sec:results}. Section \ref{sec:conclusion} presents the conclusion and outlook.
	
	
	\section{System and signal model}
	\label{sec:model}
	The procedure of a radar-based \gls{CWS} can be divided into three steps: detecting object, estimating parameters and making decisions \cite{HEIDENREICH20133400}, as shown in Fig.~\ref{fig:diagram}. More detailed explanations are given as follows:
	\begin{enumerate}
		\item Object detection: The system tells whether there is an object in the front according to the received signal of radar. \Copy{object_detection}{A common method is to compare the amplitude of signal after \gls{MF} with a certain threshold \cite{kay1993fundamentals}. If the amplitude is larger than the threshold, the system tells that an object is detected. The performance of object detection mainly depends on the \gls{SNR}.}
		\item Parameter estimation: If an object in the front is detected, the system estimate parameters of this object, e.g., range and relative velocity, with the received signal. \Gls{MLE} is usually used for this purpose. The bandwidth and duration of waveform influence the accuracy of measurements.
		\item Decision making: Based on the parameter estimates, the system decides whether the current situation is safe (negative hypothesis $\mathcal{H}_0$) or threatening (positive hypothesis $\mathcal{H}_1$) following a decision rule. If the situation is decided to be threatening ($\mathcal{H}_1$), the system will give a warning to the driver and the driver can response (e.g., brake) accordingly. Otherwise ($\mathcal{H}_0$), the system will not intervene on the driver's behaviour. The rule of making decisions will be discussed in Section \ref{sec:performance}.
	\end{enumerate}
	
	
	\begin{figure}[t]
		\begin{minipage}[b]{1.0\linewidth}
			\centering
			\centerline{\includegraphics[width=8.0cm]{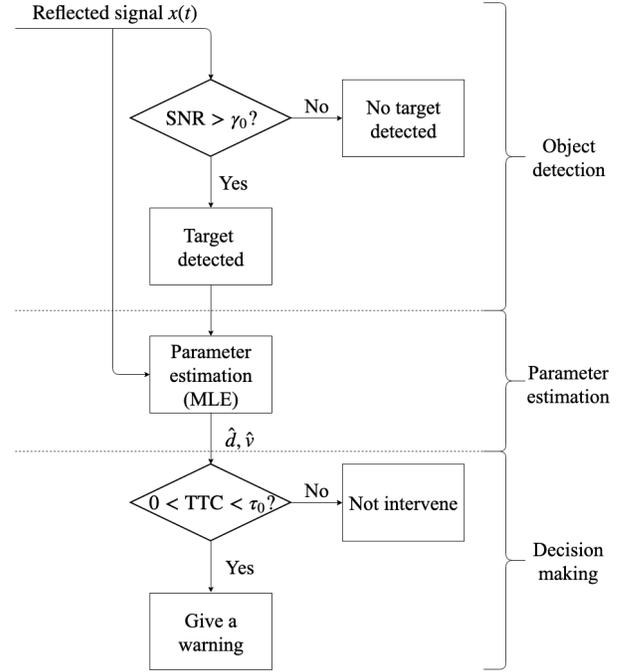}}
		\end{minipage}
		\caption{A diagram of a radar-based \gls{CWS}. The two dashed lines partition the three steps: object detection, parameter estimation and decision making.}
		\label{fig:diagram}
		\vspace{-0.4cm}
	\end{figure}
	
	
	In this paper, we assume that an object is already detected in the front and focus on the second and third step. These two steps are analyzed in this section and Section \ref{sec:performance}, respectively.
	
	\Copy{purpose_distribution}{Our purpose is to design waveform parameters that optimize \gls{CWS} performance, which relies on the distribution of radar's estimation errors. 
		Therefore, we study the error distribution using the following steps}: First, we present the signal model of an \gls{FMCW} radar in Section \ref{subsec:signal_model}. Then, we derive the \gls{CRLB} for range and velocity estimation in Section \ref{subsec:crlb}, which facilitates the analysis on the error distribution in Section \ref{subsec:distribution}.
	
	\vspace{-0.3cm}
	\subsection{Signal model}
	\label{subsec:signal_model}
	
	We analyze an \gls{FMCW} radar using linear modulation. \Copy{no_int}{It is assumed that the radars are operating in the non-overlapped coordination scheme, so no interference exists among them.} 
	
	
	The \gls{FMCW} waveform is divided into $M$ chirps and the pulse repetition interval (PRI) is $T_0$, resulting in the duration as $T = MT_0$. Within each chirp, the frequency of the waveform sweeps linearly from the initial frequency $f_0$ to $f_0+W$, where $W$ is the waveform bandwidth. We define the chirp rate as $\kappa = W/T_0$ and the amplitude is a real number $A$. Thus, the $m$-th ($m=0,1,\dots,M-1$) chirp can be written as
	\begin{equation} \label{eqn:transmit}
	\begin{aligned}
	s_m(t) = A e^ {j2\pi\left(f_0 t_m +\frac 12 \kappa t_m^2 \right)},  mT_0\leq t < (m+1)T_0,
	\end{aligned}
	\end{equation}
	where we define $t_m = t-(m-1)T_0$.
	
	The transmitted waveform is reflected by $K$ scatterers. The echo from the $k$-th $(k=1,2,\dots,K)$ scatterer is \cite{stove1992linear}
	\vspace{-0.1cm}
	\begin{equation} \label{eqn:echo}
	x_{k,m}(t) = \alpha_k s_m\left(t-{2(d_k+ v_k t)}/{c}\right),
	\vspace{-0.1cm}
	\end{equation}
	where $d_k$ and $v_k$ are the range and relative velocity of the $k$-th scatterer, and $\alpha_k$ is the complex reflection intensity. The speed of light is denoted by $c$. The received signal within the $m$-th chirp, $x_m(t)$, is the sum of $K$ echoes and noise, i.e.,
	\vspace{-0.1cm}
	\begin{equation} \label{eqn:receive_signal}
	x_m(t) = \sum_{k=1}^{K} x_{k,m}(t)+w_m(t),
	\vspace{-0.1cm}
	\end{equation}
	where $w_m(t)$ is the additive white Gaussian noise with power spectral density of $N_0$. \Copy{clutter}{Note that the clutter is a key component in automotive radars. It is generated by the reflection of road surfaces and man-made structures on roads and 
	can be treated as several stationary scatterers \cite{hyun2010two}, meaning that it is included into the multi-scatterer signal model \eqref{eqn:receive_signal}.} 
	
	Upon mixing $x_m(t)$ with a local oscillator signal $s_{m,0}(t)=s_m(t)/A$ and low-pass filtering, the intermediate frequency (IF) signal $y_m(t)$ is obtained. With narrow-band and low-velocity assumption, the IF signal can be written as \cite{stove1992linear}
	\vspace{-0.1cm}
	\begin{equation} \label{eqn:if_signal}
	\begin{aligned}
	y_m(t) \!=\! \sum_{k=1}^K A\alpha_k^* e^{j\phi_k}e^{j2\pi  f_{v,k} T_{0}m+(f_{d,k}+f_{v,k})t_m}\!+\!w'_m(t).
	\end{aligned}
	\vspace{-0.1cm}
	\end{equation}
	The beat frequency $f_{d,k}$ and the Doppler shift $f_{v,k}$ are
	\vspace{-0.1cm}
	\begin{equation} \label{eqn:fd_fv}
	f_{d,k}={2 \kappa d_k }/{c},
	f_{v,k} = {2f_0 v_k }/{c}.
	\vspace{-0.1cm}
	\end{equation}
	The phase $\phi_k = {4\pi f_0 d_k}/{c} $ is time-invariant. The operator $^*$ yields the complex conjugate of a number. The noise becomes $w_m'(t) = s_{m,0}(t)w_m^*(t)$. The IF signal $y_m(t)$ is then sampled at rate $f_s$, 
	yielding the discrete-time signal as follows:
	\vspace{-0.2cm}
	\Copy{y_nm}{
		\begin{equation} \label{eqn:y_nm}
		\begin{aligned}
		y[n,\!m] 
		= \sum_{k=1}^K \! A\alpha_k^* e^{j\phi_k} \!e^{j2\pi  \left(f_{v,k} T_{0}m + (f_{d,k}\!+\!f_{v,k})n/\!f_s\right)}
		\!+\!w'[n,\!m],
		\end{aligned}
		\end{equation}}
	where $n=0,1,\dots,N-1$ denotes the fast-time index. The total number of samples within one chirp is $N=f_s T_0$. Here, we assume that $f_s T_0$ is an integer. The variance of the noise $w'[n,m]$ is $\sigma_w^2 = f_s N_0$.
	The following manipulations are conducted on \eqref{eqn:y_nm}: \Copy{mani_y}{The coefficients are merged as $A_k'=A\alpha_k^* e^{j\phi_k}$. We use $b_k$ and $\psi_k$ to denote the amplitude and phase of $A_k'$, i.e., $A_k' = b_k e^{j\psi_k}$.} 
	We define the two-dimensional frequencies
	\vspace{-0.1cm}
	\begin{equation} \label{eqn:freq_2d}
	f_{1,k}=(f_{d,k}+f_{v,k})/f_s,
	f_{2,k}=f_{v,k} T_0.
	\end{equation}
	\vspace{-0.1cm}
	Thus, the samples shown in \eqref{eqn:y_nm} can be rewritten as
	\vspace{-0.1cm}
	\begin{equation} \label{eqn:y_nm2}
	y[n,m] = \sum_{k=1}^K b_k e^{j\psi_k} e^{j2\pi  \left(f_{1,k}n+f_{2,k}m\right)}+w'[n,m].
	\vspace{-0.1cm}
	\end{equation}	
	The deterministic parameters to be estimated are formulated as $\boldsymbol{\xi}=[\boldsymbol{\xi}_1^T,\boldsymbol{\xi}_2^T,\dots,\boldsymbol{\xi}_K^T]^T$, where $\boldsymbol{\xi}_k=[b_k,\psi_k,f_{1,k},f_{2,k}]^T$. Here, the operator $^T$ yields the transpose of a matrix. 
	After $\boldsymbol{\xi}$ is estimated, we further obtain the estimate of $\bm{\theta}_k = [d_k,v_k]^T$ according to (\ref{eqn:fd_fv}) and (\ref{eqn:freq_2d}). \Copy{phase_d}{Note that although $\phi_k$ contains the range $d_k$, it cannot be used to estimate $d_k$ since the complex reflection intensity $\alpha_k$ also introduces an unknown phase, which cannot be separated from $\phi_k$. Thus, we use the frequencies $f_{1,k}$ and $f_{2,k}$ to estimate $\bm{\theta}_k$.}
	
	As addressed in the beginning of this section, one common method of estimating $\bm{\theta}_k$ is \gls{MLE}. \Copy{MLE}{The process of \gls{MLE} is as follows: we first conduct two-dimensional fast Fourier transform \cite{hua1992estimating} on $y[n.m]$ and conduct detection in all the resolution cells. For each resolution cell that a target (numbered as $k$) is detected, we search the frequencies $\hat{f}_{1,k}$ and $\hat{f}_{2,k}$ within this cell such that the following expression
		\begin{equation}
		\left|\sum_{n=0}^{N-1}\sum_{m=0}^{M-1} y[n,m] e^{-j2\pi  \left(\hat{f}_{1,k}n+\hat{f}_{2,k}m\right)}\right|
		\end{equation}
		reaches its maximum. The corresponding $\hat{f}_{1,k}$ and $\hat{f}_{2,k}$ are the estimate of two-dimensional frequencies. Estimate of range and velocity, $\hat{d}_k$ and $\hat{v}_k$, is obtained based on their relationship with the two frequencies given by \eqref{eqn:fd_fv} and \eqref{eqn:freq_2d}.}
	
	\vspace{-0.3cm}
	\subsection{CRLB for range and velocity estimation}
	\label{subsec:crlb} 
	Due to the existence of noise, estimation error is inevitable. Thus, the estimate can be written as
	\vspace{-0.1cm}
	\begin{equation}
	\hat{\bm{\theta}}_k = \bm{\theta}_k + \bm{\varepsilon}_k,
	\vspace{-0.1cm}
	\end{equation}
	where $\hat{\bm{\theta}}_k = [\hat{d}_k, \hat{v}_k]^T$ denotes the estimate of range and velocity, and $\bm{\varepsilon}_k = [\varepsilon_{d,k},\varepsilon_{v,k}]^T$ denotes the error. 
	The \gls{CRLB} expresses a lower bound on the variance of unbiased estimators of a parameter \cite{kay1993fundamentals} and we use it to evaluate the estimation error. 
	The \gls{CRLB} for $\bm{\theta}_k$ is presented as follows:
	\Copy{prop1}{
		\begin{prop}
			\label{prop:crlb}
			When $M$ and $N$ are sufficiently large, errors of estimating $\bm{\theta}_k$ for different $k$ are asymptotically uncorrelated and the \gls{CRLB} for estimating $\bm{\theta}_k$ can be approximated as 
			\begin{equation} \label{eqn:crlb_theta_k}
			\begin{aligned}
			\mathbf{B}_{\theta,k}=
			\frac{3c^2}{8\pi^2 \gamma_k} \mathrm{diag}\left(\frac{1}{W^2}, \frac{1}{f_0^2 T^2}\right),			
			\end{aligned}
			\end{equation}
			where $\gamma_k = {\abs{\alpha_k}^2 A^2 T}/{N_0}$ is the \gls{SNR} after \gls{MF}.
		\end{prop}
	}
	
	\begin{proof}
		See Appendix \ref{appendix:CRLB}.	
	\end{proof}
	
	

	We use $B_{d,k}$ and $B_{v,k}$ to denote the \gls{CRLB}s, i.e.,
	\begin{equation} \label{eqn:crlb_dv}
	B_{d,k}=\frac{3c^2}{8\pi^2 \gamma_k}\frac{1}{W^2},B_{v,k}=\frac{3c^2}{8\pi^2 \gamma_k}\frac{1}{f_0^2 T^2}.
	\end{equation}

	When designing $W$ and $T$, we assume that the \gls{SNR}  $\gamma_k$ is a constant to fix the mean power of radar.
	The results in \eqref{eqn:crlb_dv} illustrate that when $\gamma_k$ is fixed, the \gls{CRLB}s of $d_k$ and $v_k$ are inversely proportional to $W^2$ and $T^2$, respectively. This implies that the accuracy of range and velocity measurement requires resources in the frequency time domain, respectively. 
	
	
	Moreover, although multiple scatterers are observed, we can consider that only one of them is critical for the system to decide whether to give a warning or not. 
	We select the critical scatterer with the smallest positive \gls{TTC}, since the \gls{CWS} depends on such a scatterer to decide whether to warn or not, as shown in Section \ref{sec:performance}. In the remainder, we assume that this scatterer is already selected and the subscript $k$ is omitted for conciseness, e.g., $B_d = B_{d,k}$.
	
	In conventional design of radar waveform, the bandwidth and duration are designed according to requirement on resolution of range and velocity, respectively. Their designs are separate.  
	Nevertheless, when the waveform TBP is limited to restrict resource occupation, \eqref{eqn:crlb_dv} implies that there is a tradeoff of accuracy between range and velocity. Our main purpose is to design the bandwidth and duration jointly to reach a tradeoff where the \gls{CWS} achieves its best performance.

	\vspace{-0.3cm}
	\subsection{Distribution of estimation errors}
	\label{subsec:distribution}
	We are now ready to derive the distribution of $\bm{\varepsilon}$, as required for analysis on \gls{CWS} performance. \Copy{distribution_theorem}{As stated in \cite[Theorem 7.3]{kay1993fundamentals}, when the \gls{SNR} is sufficiently high, the estimation error using \gls{MLE} is asymptotically distributed (for large $M$ and $N$) as a zero-mean Gaussian variable whose covariance is the CLRB}. \Copy{expression_sigma}{Thus, we assume that $\bm{\varepsilon}$ follows zero-mean Gaussian distribution with 
	covariance $\bm{\Sigma}_{\bm{\varepsilon}} = \diag(\sigma^2_{d},\sigma^2_{v})$,
	where
	\begin{equation} \label{eqn:sigma_dv}
	\sigma^2_{d} = B_d = \frac{3c^2}{8\pi^2 \gamma}\frac{1}{W^2},\sigma^2_{v} = B_v =\frac{3c^2}{8\pi^2 \gamma}\frac{1}{f_0^2 T^2}.
	\end{equation}}
	\Copy{distribution_test}{\Copy{error_MLE}{To test this assumption, we generate observations according to \eqref{eqn:y_nm} and use \gls{MLE} to get samples of estimates. Two results are drawn from these samples: First, when the \gls{SNR} is greater than 10 dB, which is a common case in practice \cite{forstner200877ghz}, the covariance of range and velocity errors is close to their \gls{CRLB}}; Second, we conduct the Kolmogorov-Smirnov test \cite{massey1951kolmogorov} on samples of $\varepsilon_d$ and $\varepsilon_v$, and the results do not reject the null hypothesis that the samples are drawn from the assumed Gaussian distribution at the $5\%$ significance level, which verifies that the error follows a Gaussian distribution.} The above results construct a preliminary verification about our assumption. It will be further checked through simulations in Section \ref{sec:results}.

	\section{Performance of collision warning system}
	\label{sec:performance}
	\Copy{TWDL_motivation}{After the range and velocity are estimated, the \gls{CWS} decides whether the current situation is safe or threatening using a certain rule and takes corresponding actions. In this step, estimation errors of radar may lead to wrong decisions, which can be divided into two types: The first is triggering an alarm in a safe situation, i.e., deciding $\mathcal{H}_0$ as $\mathcal{H}_1$, called a false alarm; The second regards a threatening case as safe, i.e.,  deciding $\mathcal{H}_1$ as $\mathcal{H}_0$, called a miss \cite{kay1993fundamentals}. In a \gls{CWS}, we hope to reduce both false alarm and miss. However, when the radar performance is fixed, there is a tradeoff between these two types of wrong decisions \cite{lee2005evaluation}. Therefore, we propose \gls{TWDL} that balances wrong decisions under all possible situations and further propose \gls{MTWDL} as the performance metric of CWS based on \gls{TWDL}.} Then, we establish the relationship between the CWS performance metric and radar performance. 
	
	The remainder is arranged as follows: First, we set a criterion to partition $\mathcal{H}_0$ and $\mathcal{H}_1$ as well as test statistics 
	in Section \ref{subsec:hypo_rule}. Second, we introduce \gls{TWDL} as the performance metric of decision rule in Section \ref{subsec:metric}. Third, we further introduce \gls{MTWDL} as the performance metric of CWS and analyze its relationship with radar performance in Section \ref{subsec:property_twdl}.
	
	\vspace{-0.3cm}
	\subsection{Two hypotheses and test statistics}
	\label{subsec:hypo_rule}
	First, we need a criterion to partition safe and threatening situations (also called $\mathcal{H}_0$ and $\mathcal{H}_1$ hypothesis) based on the ground truth $\bm{\theta}$ \cite{lee2005evaluation,lee2016real}. Later, we will extend to the case when the ground truth is unknown. 
	In this respect, the \gls{TTC} is proven effective for partition \cite{lee2005evaluation,minderhoud2001extended} and it is widely used in real systems \cite{jermakian2017effects,phillips2019real}. The \gls{TTC} is defined as $\tau = -{d}/{v}$. 
	%
	Based on the \gls{TTC}, hypotheses $\mathcal{H}_1$ and $\mathcal{H}_0$ are given by \cite{lee2005evaluation}
	\begin{equation} \label{eqn:cw}
	\begin{aligned}
	\mathcal{H}_1 &: 0<\tau<\tau_0,\\
	\mathcal{H}_0 &: \mbox{else}.
	\end{aligned}
	\end{equation}
	That is, a threatening situation occurs when the vehicle is approaching some object 
	and the \gls{TTC} is less than a threshold $\tau_0$. Typical values of $\tau_0$ are from three to five seconds in application \cite{minderhoud2001extended}. 
	Additionally, an equivalent form of (\ref{eqn:cw}) is
	\begin{equation} \label{eqn:hyp_final}
	\begin{aligned}
	\mathcal{H}_1 &: d+\tau_0 v<0,\\
	\mathcal{H}_0 &: d+\tau_0 v\geq 0.
	\end{aligned}
	\end{equation}
	
	\Copy{glrt}{In practice, the real value of $\bm{\theta}$ cannot be known exactly due to the noise. The \gls{CWS} makes a decision with the estimate $\hat{\bm{\theta}} = [\hat{d}, \hat{v}]^T$, which contains error $\bm{\varepsilon}$. 
		One common method for decision with unknown parameters is the \gls{GLRT}, which usually performs well \cite{kay1998fundamentals}.} 
	\Copy{GLRT_result}{According to Appendix \ref{appendix:glrt}, the \gls{GLRT} rule for the \gls{CWS} is
	\begin{equation} \label{eqn:glrt_rule}
	T_G(\hat{\bm{\theta}}) \mathop  \gtrless \limits_{\mathcal{H}_1}^{\mathcal{H}_0}
	\lambda,
	\end{equation}
	where the test statistic is
	\begin{equation} \label{eqn:T_G}
	\begin{aligned}
	T_G(\hat{\bm{\theta}})\! =\! \left\{
	\begin{aligned}
	\hat{d}+\tau_0 \hat{v}, &\  \hat{d}>0,\hat{d}\geq \frac{\sigma_{d}^2}{\sigma_{v}^2 \tau_0}\hat{v}, \\
	\sqrt{\sigma_{d}^2\!+\!\tau_0^2 \sigma_{v}^2}\sqrt{\frac{\hat{d}^2}{\sigma_{d}^2}\!+\!\frac{\hat{v}^2}{\sigma_{v}^2}}, & \  \hat{d}>0,\hat{d}< \frac{\sigma_{d}^2}{\sigma_{v}^2 \tau_0}\hat{v}.
	\end{aligned}
	\right.
	\end{aligned}
	\end{equation}
}
	
	According to (\ref{eqn:T_G}), there are two expressions of the test statistic $T_G(\hat{\bm{\theta}})$ depending on the region that $\hat{\bm{\theta}}$ lies in. For the simplicity of analysis, we also consider a test statistic using the first expression 
	appearing in $T_G(\hat{\bm{\theta}})$, given by
	\begin{equation}
	T_A(\hat{\bm{\theta}})=\hat{d}+\tau_0 \hat{v}.
	\end{equation}
	Here, we make an explanation on why we construct the new statistic with the first expression in \eqref{eqn:T_G}. According to \eqref{eqn:hyp_final}, $d+\tau_0 v=0$ is the boundary between safe and threatening situations. An intuition is that a situation with estimation $\hat{\bm{\theta}} = [\hat{d},\hat{v}]^T$ lying near the boundary has high potential of making a wrong decision and such situations are critical for the performance of \gls{CWS}. Most of these critical points $\hat{\bm{\theta}}$ lie in the first region shown in \eqref{eqn:T_G} and the approximate test statistic $T_A(\hat{\bm{\theta}})$ equals $T_G(\hat{\bm{\theta}})$ for these points. Moreover, for the critical points that lie in the second region shown in \eqref{eqn:T_G}, the difference between $T_A(\hat{\bm{\theta}})$ and $T_G(\hat{\bm{\theta}})$ is small. Therefore, $T_A(\hat{\bm{\theta}})$ can be a good approximate for $T_G(\hat{\bm{\theta}})$. This will be further verified through numerical results in Section \ref{sec:results}. 
	
	We use $T_A(\hat{\bm{\theta}})$ to make decisions with the following rule:
	\begin{equation} \label{eqn:app_rule}
	T_A(\hat{\bm{\theta}}) \mathop  \gtrless \limits_{\mathcal{H}_1}^{\mathcal{H}_0}
	\lambda,
	\end{equation}
	In the remainder, we call (\ref{eqn:glrt_rule}) the \gls{GLRT} rule and call (\ref{eqn:app_rule}) the approximate rule. \Copy{rule_threshold}{Both rules depend on a threshold $\lambda$. It is a tunable variable and selected to optimize the decision rule. For this purpose, we introduce a performance metric for decision rules in Section \ref{subsec:metric}, which enables us to find the optimal $\lambda$.}
	
	\vspace{-0.3cm}
	\subsection{Performance metric of decision rule}
	\label{subsec:metric}
	
	Since the ground truth $\bm{\theta}$ is unknown to the \gls{CWS}, wrong decisions happen using a decision rule. 
	In order to evaluate the performance of a decision rule, we introduce a metric called \gls{TWDL}, which we describe in the following.
	
	\Copy{TWDL}{\Copy{pw_distribution}{For a given pair of true values $d$ and $v$, we use $P_\mathrm{w}(d,v)$ to denote the probability of making a wrong decision, which depends on two hypotheses \eqref{eqn:hyp_final}, decision rule and error distribution.} Moreover, we use \gls{PWDL}, denoted by $u(d, v)$, to quantify the loss of a wrong decision. It is manually designed to reflect the severity of a wrong decision under different situations. The \gls{PWDL} is a function of $d$ and $v$, and is irrelevant of radar performance and decision rule, meaning that $u(d,v)$ is irrelevant of the decision threshold $\lambda$ and the estimation error variances $\sigma^2_d$ and $\sigma^2_v$. We integrate the product of $P_\mathrm{w}(d,v)$ and $u(d, v)$ over a domain $D$, which contains all possible pairs of $d$ and $v$ in practice. The integral is the defined \gls{TWDL}, given by
	\begin{equation} \label{eqn:twdl}
	U = \iint_D  u(d,v) P_\mathrm{w}(d,v)   \mathrm{d}d \mathrm{d}v.
	\end{equation}	
	Through integration, the \gls{TWDL} synthesizes wrong decisions under all situations within $D$ and balances them according to respective probability and loss, i.e., $P_\mathrm{w}(d,v)$ and $u(d,v)$. The \gls{TWDL} is used to quantify the performance of a decision rule.}
	
	Moreover, simulations in Section \ref{subsec:comp_rule} will show that the performance of the GLRT and approximate rules is pretty close. Our following theoretical analysis on the system performance is based on the approximate rule for simplicity.
	
	With the hypotheses \eqref{eqn:hyp_final} and approximate decision rule \eqref{eqn:app_rule}, we derive the probability $P_\mathrm{w}(d,v)$ analytically. 
	For instance, when $d+\tau_0 v$ is non-negative, representing a safe situation $\mathcal{H}_0$, $P_\mathrm{w}(d,v)$ is the probability that a warning is given according to the rule \eqref{eqn:app_rule} with estimate $\hat{\bm{\theta}} = [\hat{d}, \hat{v}]^T$. Its derivation is as follows: 
	\begin{equation} \label{eqn:pw_derive}
	\begin{aligned}
	P_\mathrm{w}(d,v) & = \mathrm{Pr}\left(\hat{d}+\tau_0 \hat{v}<\lambda; d,v\right)\\
	& = \mathrm{Pr}\left((d+\varepsilon_d)+\tau_0 (v+\varepsilon_v)<\lambda; d,v\right)\\
	& = \mathrm{Pr}\left(-\frac{\varepsilon_d+\tau_0 \varepsilon_v}{\sigma_Z} > \frac{d+\tau_0 v-\lambda}{\sigma_Z}; d,v\right)\\
	& = Q\left(({d+\tau_0 v-\lambda})/{\sigma_{Z}}\right),
	\end{aligned}
	\end{equation}
	where $\sigma^2_{Z}$ is the variance of $\varepsilon_d+\tau_0 \varepsilon_v$, which is a weighted sum of range and velocity error variances, given by
	\begin{equation} \label{eqn:error}
	\sigma^2_{Z} = \sigma_{d}^2 + \tau_0^2 \sigma_{v}^2.
	\end{equation}
	We define $\sigma^2_{Z}$ as range-velocity joint error index (called error index for short). 
	The univariate Q-function $Q(x)$ is given by
	\begin{equation} \label{eqn:qfunc}
	Q(x) = \frac{1}{\sqrt{2\pi}} 
	\int_x^{\infty} 
	\exp\left(-\frac{t^2}{2}\right)
	\mathrm{d}t.
	\end{equation}
	
	Similarly, we can obtain the probability $P_\mathrm{w}(d,v)$ for negative $d+\tau_0 v$. The complete expression of $P_\mathrm{w}(d,v)$ is
	\begin{equation} \label{eqn:pw}
	P_\mathrm{w}(d,v)=\left\{
	\begin{aligned}
	Q\left(({d+\tau_0 v-\lambda})/{\sigma_{Z}}\right), & \quad d+\tau_0 v\geq 0, \\
	Q\left(-({d+\tau_0 v-\lambda})/{\sigma_{Z}}\right), & \quad d+\tau_0 v < 0.
	\end{aligned}
	\right.
	\end{equation}
	
	The result (\ref{eqn:pw}) indicates that $P_\mathrm{w}(d,v)$ depends on the decision threshold $\lambda$ and square root of error index  $\sigma_{Z}$. Therefore, when the  integral domain $D$ and manually designed \gls{PWDL} $u(d,v)$ are fixed, the \gls{TWDL} $U$ shown in (\ref{eqn:twdl}) also depends on the variables $\lambda$ and $\sigma_{Z}$. We use the notation $U(\lambda,\sigma_{Z})$ to show this dependence. 
	
	Based on the loss function $U(\lambda,\sigma_{Z})$, we will tune the parameter $\lambda$ and reveal how the radar performance influence the performance of CWS in the sequel. 
	
	\vspace{-0.3cm}
	\subsection{Relationship between CWS and radar performance}
	\label{subsec:property_twdl}
	\Copy{MTWDL}{As shown in Section \ref{subsec:hypo_rule}, the threshold $\lambda$ is tunable. For a specific $\sigma_{Z}$, we tune it to $\tilde{\lambda}(\sigma_{Z})$ to minimize \gls{TWDL}, i.e.,
	\begin{equation} \label{eqn:opt_lambda}
	\tilde{\lambda}(\sigma_{Z})= \mathop{\arg\min}_{\lambda} \  U(\lambda,\sigma_{Z}).
	\end{equation}
	The optimal threshold $\tilde{\lambda}(\sigma_{Z})$ yields an optimal decision rule in terms of TWLD. 
	We define the corresponding minimal TWDL as \gls{MTWDL}, denoted by $\tilde{U}$. That is,
	\begin{equation} \label{eqn:opt_loss}
	\tilde{U}(\sigma_{Z}) = U(\tilde{\lambda}(\sigma_{Z}),\sigma_{Z}).
	\end{equation}}
	As with \gls{TWDL}, \gls{MTWDL} also balances wrong decisions within overall situations. It is obtained by tuning the decision threshold to its optimum. Therefore, we use \gls{MTWDL} as the performance metric of the CWS.
	Note that $\tilde{U}$ is a function of $\sigma_{Z}$, while $\sigma_{Z}$ is decided by the radar performance, as shown in (\ref{eqn:error}). Therefore, the function $\tilde{U}(\sigma_{Z})$ establishes the desired relationship between performances of the system and radar. 
	
	Recall that the error index $\sigma^2_{Z}$ is the variance of $\hat{d}+\tau_0 \hat{v}$, i.e., the estimate of ${d}+\tau_0 {v}$. As shown in \eqref{eqn:hyp_final} and \eqref{eqn:app_rule}, the two hypotheses are defined based on ${d}+\tau_0 {v}$ and the system makes decision based on its estimation. One intuition is that a more accurate estimate of ${d}+\tau_0 {v}$, in terms of smaller $\sigma^2_{Z}$, leads to a system with better performance. This intuition is strictly verified by the following theorem: 
	\begin{theorem} \label{theo:mono}
		The \gls{MTWDL} $\tilde{U}(\sigma_{Z})$ is a monotonically increasing function of $\sigma_{Z}$ if the following conditions are satisfied: 
		\begin{enumerate}
			\item In the domain $D$, $u(d,v)>0$ for all $(d,v)$.
			\item The function $u(d,v)$ is Lebesgue integrable over $D$.		
			\item The domain $D$ is bounded.
		\end{enumerate} 
	\end{theorem}
	\begin{proof}
		See Appendix \ref{appendix:mono_loss}.
	\end{proof}
	
	In Theorem \ref{theo:mono}, we impose mild conditions on the manually designed function $u(d,v)$, so it holds for a wild range of $u(d,v)$. It is implied that a smaller $\sigma^2_{Z}$  yields better \gls{CWS} performance in terms of \gls{MTWDL}. 
	Moreover, \eqref{eqn:sigma_dv} indicates that when  waveform \gls{TBP} is limited, there is a tradeoff of accuracies between range and velocity.
	Therefore, we should design the radar waveform parameters in the aim of minimizing $\sigma^2_{Z}$, such that the tradeoff that optimizes system performance is achieved. The quantitative relationship between the radar and system performance revealed in this theorem inspires the waveform design method presented in the following section.

	\section{Optimized waveform parameter design}
	\label{sec:design}
	According to Theorem \ref{theo:mono}, in order to achieve the best system performance, we optimize the waveform parameters $W$ and $T$ by minimizing the error index. For this purpose, we formulate the parameter optimization problem and analyze its results in Section \ref{subsec:opt_form}; We compare the optimized parameters with conventional designs for automotive radars in Section \ref{subsec:comp}.
	
	\vspace{-0.3cm}
	\subsection{Formulation of parameter optimization}
	\label{subsec:opt_form} 
	First, we substitute (\ref{eqn:crlb_dv}) into (\ref{eqn:error}). Hence, the error index depends on the bandwidth $W$ and duration $T$ as follows: 
	\begin{equation} \label{eqn:sigma_waveform}
	\sigma^2_{Z} = 
	\sigma_{d}^2 + \tau_0^2 \sigma_{v}^2 =\frac{3c^2}{8\pi^2 \gamma}\left(
	\frac{1}{W^2}+\frac{\tau_0^2}{f_0^2}\frac{1}{T^2}
	\right).
	\end{equation}
	Therefore, the parameter design problem can be given by
	\begin{equation} \label{eqn:optimization}
	\begin{aligned}
	&\min_{W,T} &&\sigma^2_{Z}  = \frac{3c^2}{8\pi^2 \gamma}\left(
	\frac{1}{W^2}+\frac{\tau_0^2}{f_0^2}\frac{1}{T^2}
	\right),\\
	&\mbox{s.t.}&&(\mbox{C1})\quad 0 < W \leq W_{\mathrm{max}}, 0 < T \leq T_{\mathrm{max}},\\
	& &&(\mbox{C2})\quad WT \leq S.
	\end{aligned}
	\end{equation}
	Here, we impose the following constraints on the optimization:
	\begin{enumerate}
		\item[(C1)]  The variables $W$ and $T$ are bounded. The upper bounds are interpreted as follows: First, the bandwidth available to all automotive radars is $W_{\max}$, which upper bounds $W$. 
		Second, there is a requirement on the update rate of measurement
		, and the duration of the waveform cannot exceed the update period $T_{\max}$.
		\item[(C2)] The \gls{TBP} of the waveform $WT$ is no greater than a \gls{TBP} limit $S$, in order to limit resource occupation of the radar. In practice, $S$ is set according to the number of coexisting radars $N_r$, such that $N_r S\leq W_{\mathrm{max}} T_{\mathrm{max}}$.
	\end{enumerate}
	
	Moreover, in Section \ref{sec:model}, possible values of $T$ are $M T_0$, where $M$ is an integer. In the optimization problem \eqref{eqn:optimization}, however, we treat $T$ as a continuous variable since $T_0$ is far smaller than $T$ and the change of $T$ over $M$ is relatively small.
	
	To solve \eqref{eqn:optimization}, we first assume that  $W_{\mathrm{max}}$ and $T_{\mathrm{max}}$ are reasonably large, such that $W$ and $T$ take their optima when the constraint (C1) is inactive. 
	Using the inequality of arithmetic and geometric means, we have
	
	\begin{equation}\label{eqn:sigma_opt}
	\sigma^2_{Z} \geq \frac{3c^2}{8\pi^2 \gamma}\frac{2\tau_0}{f_0 WT}\geq \frac{3c^2\tau_0}{4\pi^2\gamma f_0 S}= \sigma^2_{Z,\mathrm{opt}},
	\end{equation}
	where $\sigma^2_{Z,\mathrm{opt}}$ denotes the optimal error index. The equality holds if and only if
	\begin{equation}\label{eqn:W_T_opt}
	W = W_\mathrm{opt}= \sqrt{{f_0 S}/{\tau_0}},  T=T_\mathrm{opt}= \sqrt{{\tau_0 S}/{f_0}}.
	\end{equation}
	The solution \eqref{eqn:W_T_opt} is optimal as long as $W_\mathrm{opt}\leq W_\mathrm{max}$ and $T_\mathrm{opt}\leq T_\mathrm{max}$.
	In the case that $W_\mathrm{opt}> W_\mathrm{max}$ or $T_\mathrm{opt} > T_\mathrm{max}$ (At most one of the two situations may happen since we require $S\leq W_{\mathrm{max}} T_{\mathrm{max}}$), the corresponding solution is
	\begin{equation}
	\begin{aligned}
	&W = W_{\mathrm{max}}, T = S/W_{\mathrm{max}}, &&\mbox{ if } W_\mathrm{opt}> W_\mathrm{max},\\
	&W = S/T_{\mathrm{max}}, T = T_{\mathrm{max}}, &&\mbox{ if } T_\mathrm{opt}> T_\mathrm{max}.
	\end{aligned}
	\end{equation}
	In practice, $W_\mathrm{opt}\leq W_\mathrm{max}, T_\mathrm{opt}\leq T_\mathrm{max}$ is the general case and will be analyzed in the remainder.
	
	The following two observations are made about the optimization results in \eqref{eqn:sigma_opt} and \eqref{eqn:W_T_opt}: First, when the other parameters are fixed, the optimal error index $\sigma^2_{Z,\mathrm{opt}}$ is inversely proportional to the \gls{SNR} $\gamma$ and \gls{TBP} $S$. Second, the ratio of optimal bandwidth and duration is $W_{\mathrm{opt}}/T_{\mathrm{opt}}=f_0/\tau_0$, which is invariant of $\gamma$ and $S$. This ratio reflects the tradeoff between bandwidth and duration when the \gls{TBP} is limited. 
	
	\vspace{-0.3cm}
	\subsection{Comparison between optimized and conventional design}
	\label{subsec:comp}
	In this section, we compare the  waveform optimization strategy in Section~\ref{sec:design} with some existing approaches in the terms of CWS performance and  resource efficiency. 
	
	\Copy{conventional_design}{We consider some state-of-the-art waveform designs for automotive radars \cite{russell1997millimeter, hasch2012millimeter,sit2018bpsk,yi201924}.  
		The routine of parameter design is as follows: First, requirements on range resolution $\Delta_d$ and velocity resolution $\Delta_v$ are proposed empirically. Then, the designed bandwidth $W_{\mathrm{con}}$ and duration $T_{\mathrm{con}}$ are given by
		\begin{equation} \label{eqn:W_T_con}
		W_{\mathrm{con}} = {c}/({2\Delta_d}),  T_{\mathrm{con}} = {c}/({2 f_0 \Delta_v}).
		\end{equation}
	}
	The corresponding waveform \gls{TBP} is $S_{\mathrm{con}} = {c^2}/({4f_0 \Delta_d \Delta_v})$.
	To facilitate the comparison with our method, we calculate the error index introduced for the convensional approach, given by 
	\begin{equation} \label{eqn:error_con}
	\sigma^2_{Z,\mathrm{con}} =\frac{3c^2\tau_0}{8\pi^2\gamma f_0 S_{\mathrm{con}}} \left(\frac{\Delta_d}{\tau_0\Delta_v}+\frac{\tau_0\Delta_v}{\Delta_d}\right).
	\end{equation}
	In these methods, the bandwidth and duration of the waveform are designed separately. The issue of limited electromagnetic resources and the tradeoff of estimation accuracy between range and velocity are barely considered. 

	We now compare between the proposed and the conventional waveform design methods. In particular, we consider the parameters of error index, \gls{TBP} and \gls{SNR}, representing the CWS performance, electromagnetic resource efficiency and energy efficiency, respectively. We analyze the influence on one parameter of the design approaches when the rest parameters are fixed. 
	Details of comparisons are as follows:
	\subsubsection{Comparison on error index with \gls{SNR} and \gls{TBP} fixed}
	When the \gls{SNR} $\gamma$ is fixed and \gls{TBP} is set as $S_{\mathrm{con}}$, the results of optimized design are obtained using \eqref{eqn:W_T_opt} and \eqref{eqn:sigma_opt}. Further, we compare the error indices of two designs as follows:
	\begin{equation}\label{eqn:ratio_error}
	\frac{\sigma^2_{Z,\mathrm{opt}}}{\sigma^2_{Z,\mathrm{con}}} = \frac{2\tau_0\cdot \Delta_d/\Delta_v}{\tau_0^2 + (\Delta_d/\Delta_v)^2}.
	\end{equation}
	A deduction from \eqref{eqn:ratio_error} is that $\sigma^2_{Z,\mathrm{opt}}<\sigma^2_{Z,\mathrm{con}}$ as long as $\tau_0 \neq \Delta_d/\Delta_v$. That is, the optimized design reduces the error index and further improves the \gls{CWS} performance.
	To be specific, $\tau_0$ is the \gls{TTC} threshold of \gls{CWS}, whose typical values are 3 to 5 seconds (see \cite{minderhoud2001extended} and the references therein). The expression $\Delta_d/\Delta_v$ is the ratio of conventional range and velocity resolution, whose typical values are 0.7 to 1.2 seconds \cite{russell1997millimeter, hasch2012millimeter,sit2018bpsk,yi201924}. Thus,  $\tau_0$ is greater than $\Delta_d/\Delta_v$ in practice.
	\subsubsection{Comparison on \gls{TBP} with \gls{SNR} and error index fixed}
	In conventional design, after  $\Delta_d$ and $\Delta_v$ are specified, the bandwidth and duration are designed. Then, $\sigma_{Z,\mathrm{opt}}^2$ is given by \eqref{eqn:error_con}, which decides the \gls{CWS} performance. Here, we study how much \gls{TBP} is required using the optimized design to achieve the same performance as conventional design. Using \eqref{eqn:sigma_opt}, we obtain the \gls{TBP} $S_{\mathrm{opt}}$ to yield an error index equal to $\sigma_{Z,\mathrm{opt}}^2$. The ratio of two \gls{TBP}s is
	\begin{equation}\label{eqn:ratio_TBP}
	\frac{S_{\mathrm{opt}}}{S_{\mathrm{con}}} = \frac{2\tau_0\cdot \Delta_d/\Delta_v}{\tau_0^2 + (\Delta_d/\Delta_v)^2}.
	\end{equation}
	The result of this ratio is the same as \eqref{eqn:ratio_error} and we have ${S_{\mathrm{opt}}}<{S_{\mathrm{con}}}$ as long as $\tau_0 \neq \Delta_d/\Delta_v$. It means that the optimized design can reduce waveform \gls{TBP} while the system performance remains the same.
	\subsubsection{Comparison on \gls{SNR} with \gls{TBP} and error index fixed}
	We study what level of \gls{SNR} is required using the optimized design to achieve the same performance as conventional design. For a specific error index $\sigma_{Z}^2$ and \gls{TBP} $S$, we obtain the required \gls{SNR} for optimized and conventional design using \eqref{eqn:sigma_opt} and \eqref{eqn:error_con}, denoted by $\gamma_{\mathrm{opt}}$ and $\gamma_{\mathrm{con}}$. Their ratio is 
	\begin{equation}\label{eqn:ratio_SNR}
	\frac{\gamma_{\mathrm{opt}}}{\gamma_{\mathrm{con}}} = \frac{2\tau_0\cdot \Delta_d/\Delta_v}{\tau_0^2 + (\Delta_d/\Delta_v)^2}.
	\end{equation}
	Again, the ratio is the same as \eqref{eqn:ratio_error}, and $\gamma_{\mathrm{opt}}<\gamma_{\mathrm{con}}$ as long as $\tau_0 \neq \Delta_d/\Delta_v$. That is, the optimized design can achieve the same system performance as conventional design with lower transmitted power.

	To summarize, the above comparisons illustrate the advantages of optimized parameter design from three aspects, namely, performance improvement, resource efficiency and energy efficiency. These advantages are further demonstrated through simulations in the Section \ref{sec:results}.  

	\section{Numerical results}
	\label{sec:results}
	In this section, we use numerical results to verify previous assumptions and to compare the optimized and conventional parameter design. This section is divided into the following five parts: First, we introduce settings of the simulation in Section \ref{subsec:setting}. Second, we compare the \gls{GLRT} rule and approximate rule in Section \ref{subsec:comp_rule}. Third, we compare the error indices of the two designs in Section \ref{subsec:comp_error}. Fourth, we study the performance metric \gls{MTWDL} as a function of \gls{TBP} in Section \ref{subsec:comp_TBP} to compare the resource efficiency of the two designs. Finally, the \gls{MTWDL} is demonstrated as a function of \gls{SNR} in Section \ref{subsec:comp_SNR} to compare the energy efficiency. 
	
	\vspace{-0.3cm}
	\subsection{Simulation settings}
	\label{subsec:setting}	
	The system and waveform parameters are specified as follows: 
	$f_0 = 24 \mbox{ GHz} $ and $\tau_0 = 4 \mbox{ s} $ (as suggested in \cite{minderhoud2001extended})
	. The maximum bandwidth and duration, i.e., $W_{\max}$ and $T_{\max}$, are set as 500 MHz and 50 ms. 
	\Copy{two_designs}{For conventional design, the range and velocity resolution are $\Delta_d = 0.5\ \mbox{m}$ and $\Delta_v=0.6\ \mbox{m}/\mbox{s}$, as suggested in \cite{hasch2012millimeter}. The corresponding waveform parameters are $W_{\mathrm{con}} = 300 \mbox{ MHz}$ and $T_{\mathrm{con}} = 10.4 \mbox{ ms}$, and the \gls{TBP} is $S_{\mathrm{con}} = 3.125\times 10^6$. For optimized design, the waveform bandwidth and duration are designed by solving \eqref{eqn:optimization} given the \gls{TBP}.} Specifically, when the \gls{TBP} is $S_{\mathrm{con}}$, the optimized parameters are $W_{\mathrm{opt}} = 137\mbox{ MHz}$ and $T_{\mathrm{opt}} = 22.8\mbox{ ms}$.
	
	As for the \gls{TWDL}, the domain of interest $D$ is 
	$D: 0.1\mbox{ m}\leq d \leq 100\mbox{ m}, -30\mbox{ m/s}\leq v \leq 30\mbox{ m/s}$.
	We consider the \gls{PWDL} in the following forms:
	\begin{enumerate}
		\item Constant loss: The loss of false alarm and miss are two different constants, 
		i.e., 
		\begin{equation} \label{eqn:loss1}
		u(d,v)=\left\{
		\begin{aligned}
		1, & \quad d+\tau_0 v\geq 0, \\
		u_1, & \quad d+\tau_0 v < 0,
		\end{aligned}
		\right.
		\end{equation}
		where $u_1$ reflects the severity of miss relative to false alarm. Specifically, we set $u_1$ as 5 and 10, corresponding to \gls{PWDL} 1 and \gls{PWDL} 2, respectively.
		
		\item \gls{TTC}-dependent loss: Since the \gls{TTC} reflects the urgency of warning \cite{moon2009design}, the loss function $u(d,v)$ can be designed based on the \gls{TTC}. One example is
		\begin{equation} \label{eqn:loss2}
		u(d,v)=\left\{
		\begin{aligned}
		1, & \quad d+\tau_0 v\geq 0, \\
		u_2 \cdot (-v/d), & \quad d+\tau_0 v < 0.
		\end{aligned}
		\right.
		\end{equation}
		We set $u_2$ as 5 s and 10 s, corresponding to \gls{PWDL} 3 and \gls{PWDL} 4, respectively.
	\end{enumerate}
	
	%
	
	\vspace{-0.3cm}
	\subsection{Comparison between the GLRT and approximate rule}
	\label{subsec:comp_rule}
	In Section \ref{sec:performance}, we state that the performance of the \gls{GLRT} and approximate rule is close and conduct theoretical analysis based on the approximate rule. 
	Here, we use simulations to verify this statement.
	
	When we fix the \gls{TBP} $S$ as $S_\mathrm{con}$ and \gls{SNR} $\gamma$ as 20 dB, we can obtain the parameters of conventional and optimized design, as shown in Section \ref{subsec:setting}, and the variances, $\sigma^2_d$ and $\sigma^2_v$, according to \eqref{eqn:sigma_dv}. Further, the probability $P_\mathrm{w}(d,v)$ can be obtained through numerical integration for the \gls{GLRT} rule, while $P_\mathrm{w}(d,v)$ is given by \eqref{eqn:pw} for the approximate rule. Finally, the \gls{TWDL} $U$ can be obtained through the integration in \eqref{eqn:twdl} with different decision thresholds $\lambda$.
	
	For two configurations of waveform parameters and four \gls{PWDL} functions, we obtain the \gls{TWDL} using both the \gls{GLRT} and approximate rule as a function of the threshold $\lambda$ and the results are shown in Fig.~\ref{fig:loss_glrt_vs_appr}. 
	Two types of markers are used to distinguish the \gls{GLRT} rule (circle marker) and approximate rule (cross marker, abbreviated as Appr in the legends). Additionally, for each \gls{PWDL}, the \gls{TWDL} is divided by the minimal \gls{TWDL} over all $\lambda$ using the \gls{GLRT} rule for normalization. The results show that the normalized losses of these two rules are almost identical, thereby verifying our statement. 
	
	\begin{figure}[t]
		\subfigure[Conventional design: $W_{\mathrm{con}} = 300 \mbox{ MHz}$,  $T_{\mathrm{con}} = 10.4 \mbox{ ms}$.]{
			\begin{minipage}[t]{1.0\linewidth}
				\centering
				\centerline{\includegraphics[width=8cm]{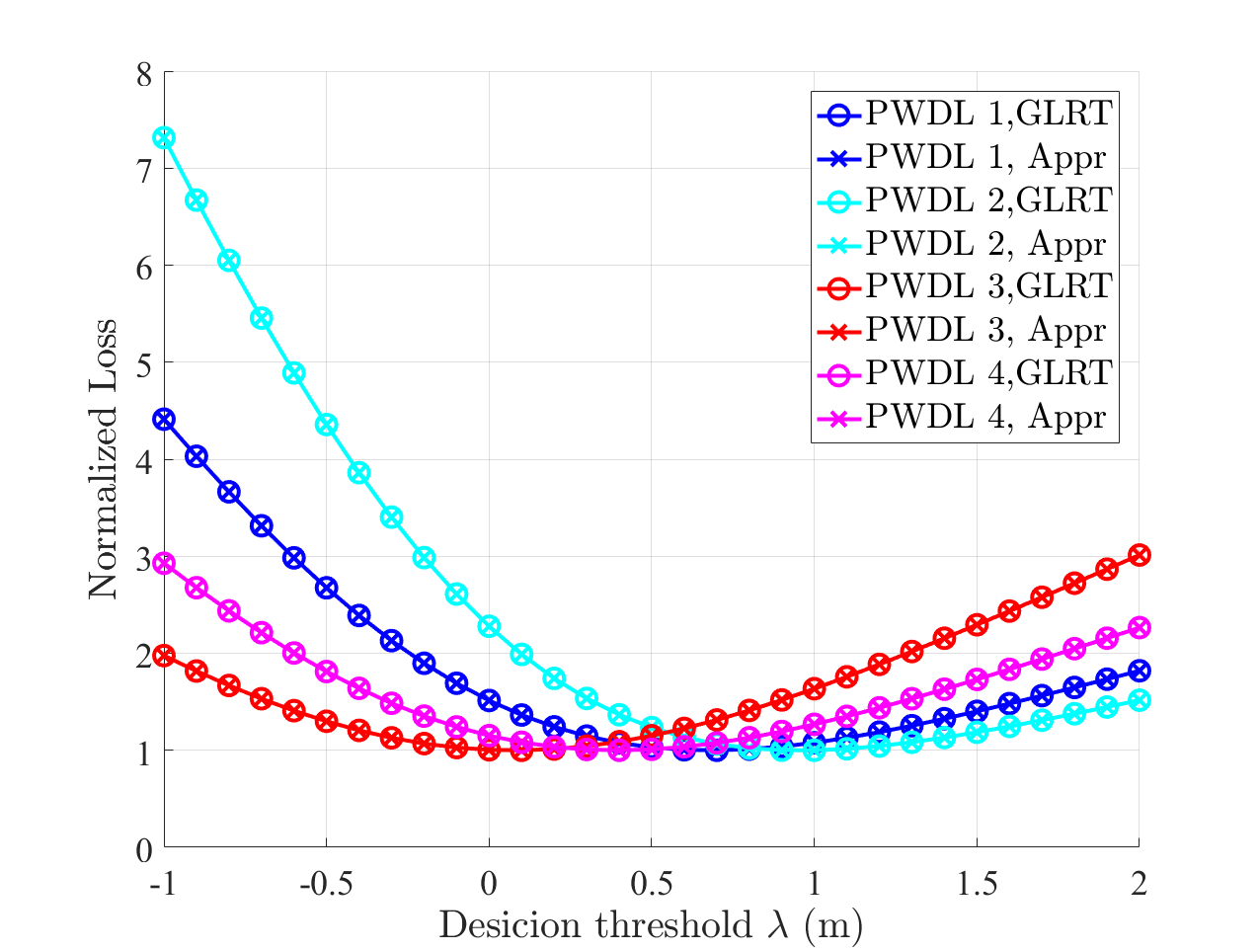}}
		\end{minipage}}
		\subfigure[Optimized design: $W_{\mathrm{opt}} = 137\mbox{ MHz}$, $T_{\mathrm{opt}} = 22.8\mbox{ ms}$.]{
			\begin{minipage}[t]{1.0\linewidth}
				\centering
				\centerline{\includegraphics[width=8cm]{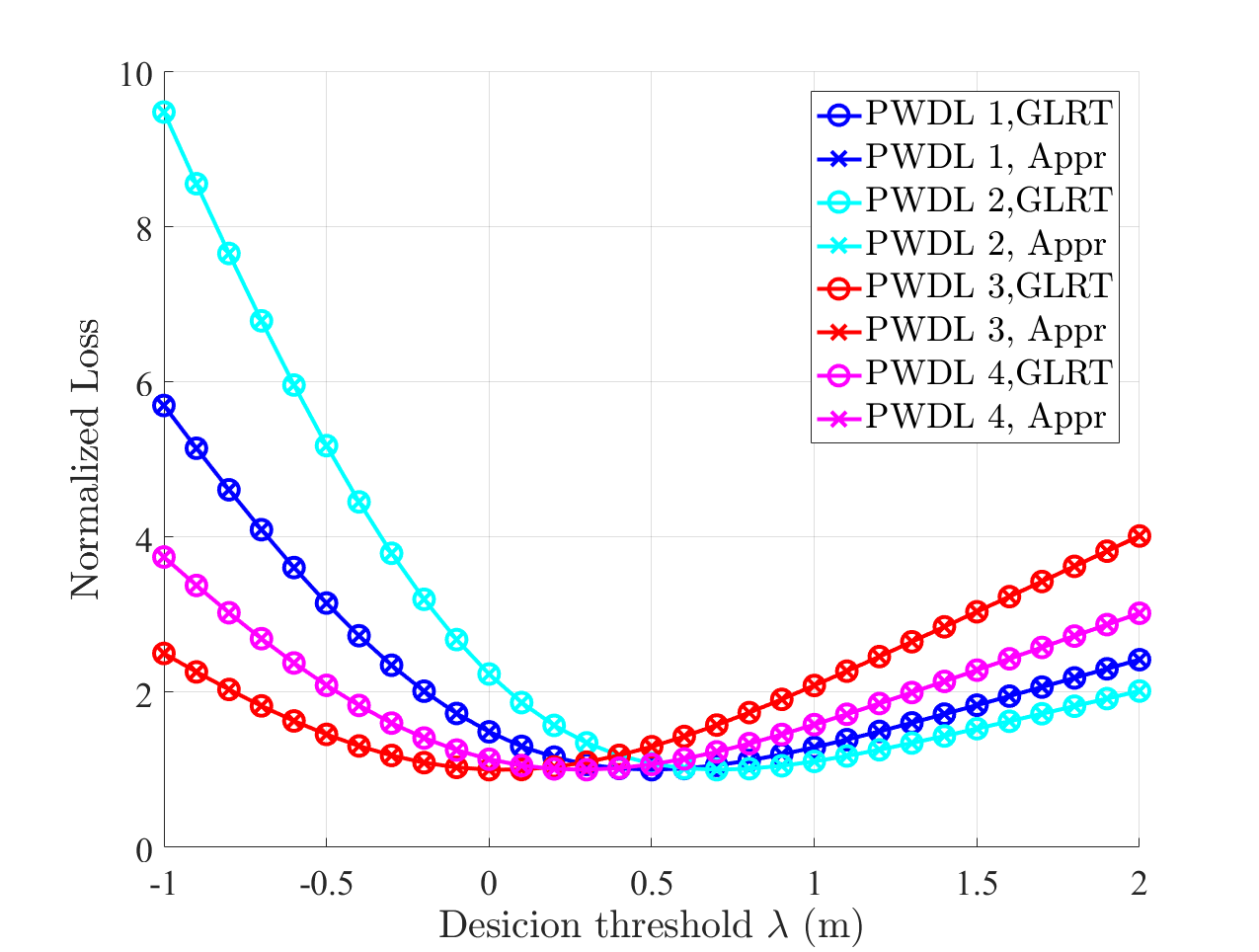}}
		\end{minipage}}
		\caption{Comparison of \gls{TWDL} using \gls{GLRT} and approximate rule with two waveform configurations and four \gls{PWDL} functions (\gls{SNR} $\gamma=20$ dB).}
		\label{fig:loss_glrt_vs_appr}
		\vspace{-0.3cm}
	\end{figure}
	
	\vspace{-0.3cm}
	\subsection{Comparison on error index} 
	\label{subsec:comp_error}
	When we fix the \gls{TBP} as $S_{\mathrm{con}}= 3.125\times 10^6$ and \gls{SNR} as $\gamma=20\mbox{ dB}$, there are different designs of bandwidth $W$ and duration $T$ that satisfy $WT=S_{\mathrm{con}}$. For each design, the theoretical result of error index $\sigma_{Z}^2$ can be obtained using \eqref{eqn:sigma_waveform}, and $\sigma_{Z}^2$ is demonstrated as a function of $W$ in Fig.~\ref{fig:error_W}. Moreover, the conventional and optimized design are given in \ref{subsec:setting} and marked as a triangle and a circle, respectively. 
	It is shown that the error index of optimized design is 4.0 dB lower than conventional design, and this difference is decided by the \gls{TTC} threshold $\tau_0$ and conventional resolution ratio $\Delta_d/\Delta_v$, as indicated in \eqref{eqn:ratio_error}. The smaller $\sigma_Z^2$ leads to a better \gls{CWS} performance, which is illustrated in the following simulations.

	\begin{figure}[htb]
		\begin{minipage}[b]{1.0\linewidth}
			\centering
			\centerline{\includegraphics[width=8cm]{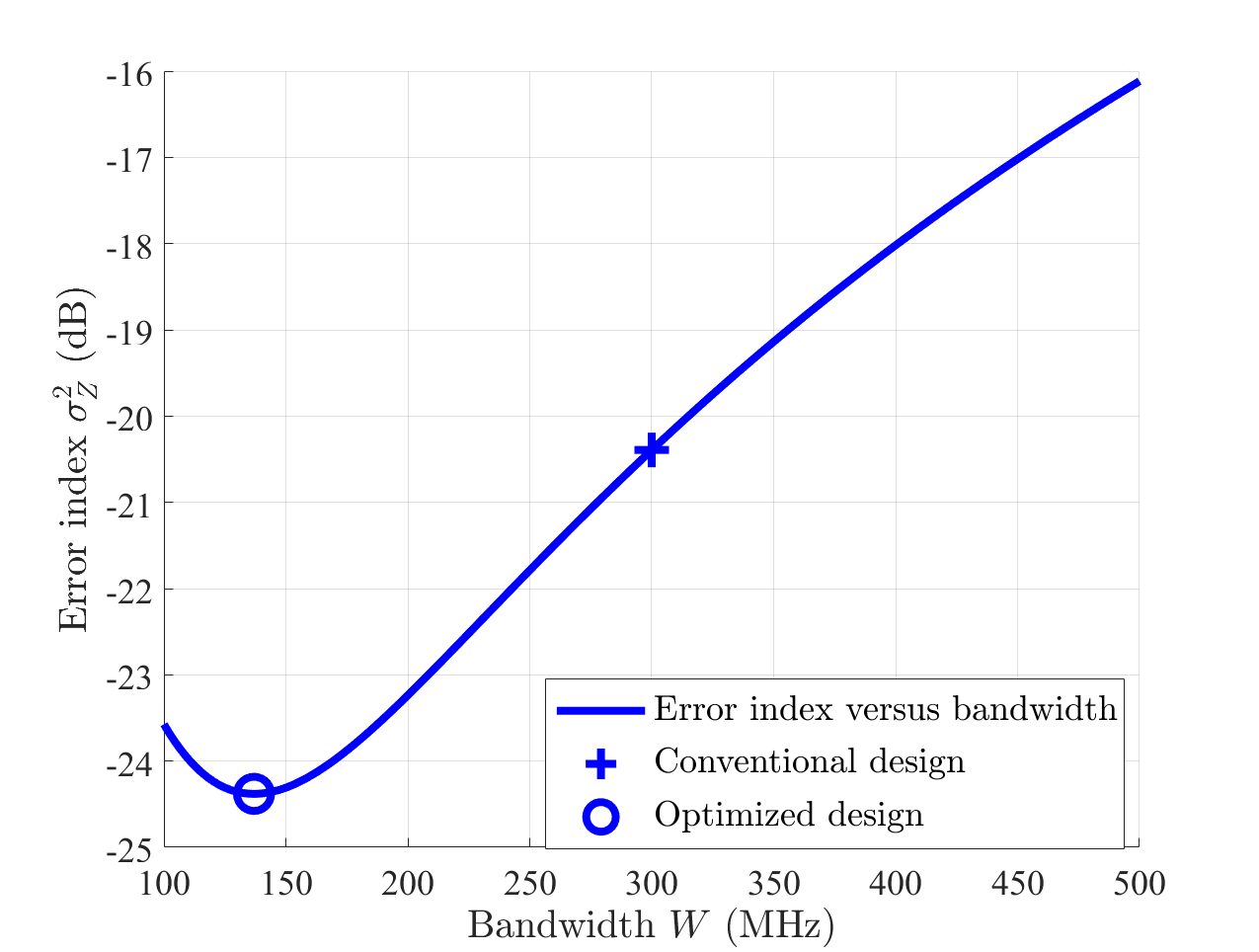}}
		\end{minipage}
		\caption{Error index $\sigma^2_{Z}$ versus bandwidth $W$ with fixed \gls{TBP}.}
		\label{fig:error_W}
		\vspace{-0.4cm}
	\end{figure}
	
	\vspace{-0.3cm}
	\subsection{Comparison on MTWDL versus TBP}
	\label{subsec:comp_TBP}
	Fig.~\ref{fig:loss_TBP} illustrates the \gls{CWS} performance metric as a function of \gls{TBP} and helps compare resource efficiencies of the two designs. Details of this figure are explained in the following.
	
	As for the conventional design (abbreviated as Conv in the legends), the bandwidth and duration are described in Section~\ref{subsec:setting}. Further, \gls{MTWDL} $\tilde{U}$ is obtained in two ways. First, we obtain the simulated results (abbreviated as Sim). \Copy{simulated_results}{We use the signal model \eqref{eqn:y_nm} to generate noisy observations and obtain samples of estimates through \gls{MLE}. These samples are used to get the probability $P_\mathrm{w}(d,v)$ for different thresholds $\lambda$ and then the \gls{TWDL} $U$. The smallest $U$ over $\lambda$ is chosen as the \gls{MTWDL} $\tilde{U}$.} These simulated results are marked as plus signs in Fig.~\ref{fig:loss_TBP}. Second, we obtain the theoretical results (abbreviated as Theo). \Copy{theo_results}{The probability $P_\mathrm{w}(d,v)$ is obtained by substituting $\sigma^2_{Z}$ into (\ref{eqn:pw}), and $U$ with specific $\lambda$ is obtained through integration described in (\ref{eqn:twdl}). Again, $\tilde{U}$ is the smallest $U$.} These theoretical results are marked as crosses in Fig.~\ref{fig:loss_TBP}. 
	
	As for the optimized design (abbreviated as Opt), with \gls{TBP} varying from $1\times 10^6$ to $1\times 10^7$, the corresponding optimized bandwidth and duration are solution of \eqref{eqn:optimization}. Similarly, we obtain both simulated and theoretical results of $\tilde{U}$, which are marked as circles and plotted as curves in Fig.~\ref{fig:loss_TBP}, respectively. 
	
	Our first observation is that our theoretical results fit the simulated results well under this \gls{SNR}. Moreover,  
	when the \gls{TBP} is fixed, the optimized design yields a better performance of the \gls{CWS} in terms of \gls{MTWDL}. 
	On the other hand, to achieve a certain performance of \gls{CWS}, the logarithm of required \gls{TBP} $\log_{10} S$ using the optimized parameter design is $0.40$ less than that using the conventional design. Equivalently, the required \gls{TBP} of optimized design is $0.40$ times that of conventional design while the CWS performance remains unchanged, which conforms to the result \eqref{eqn:ratio_TBP}. This observation indicates the remarkable advantage of the optimized design in resource efficiency. To be specific, in the non-overlapped allocation scheme, the number of radars that can simultaneously operate in the same place with no mutual interference can be increased by $150 \%$ using the optimized design while the \gls{CWS} performance remains the same as the conventional design.

	\begin{figure}[t]
		\subfigure[PWDL 1 and PWDL 2.]{
			\begin{minipage}[t]{1.0\linewidth}
				\centering
				\centerline{\includegraphics[width=8cm]{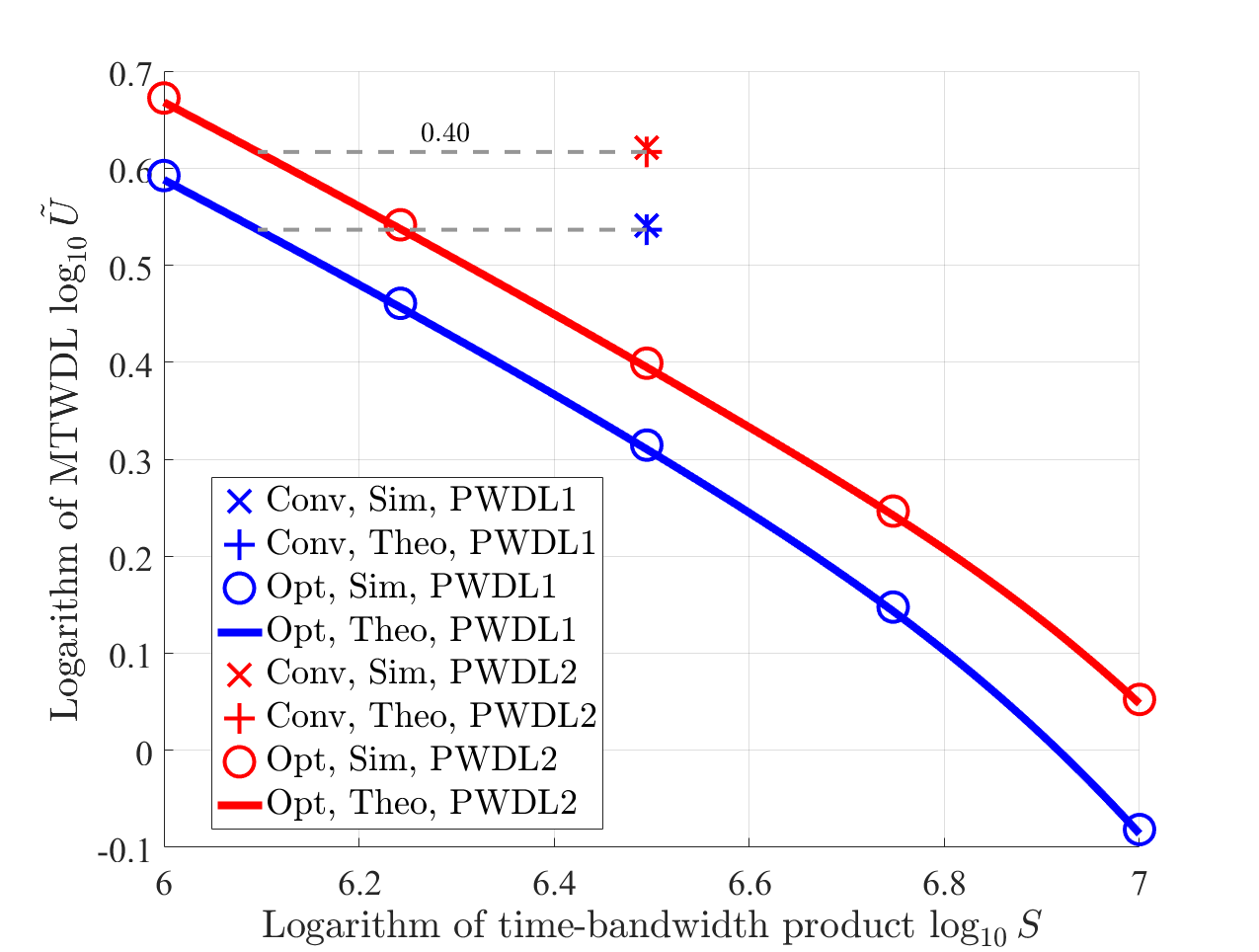}}
		\end{minipage}}
		\subfigure[PWDL 3 and PWDL 4.]{
			\begin{minipage}[t]{1.0\linewidth}
				\centering
				\centerline{\includegraphics[width=8cm]{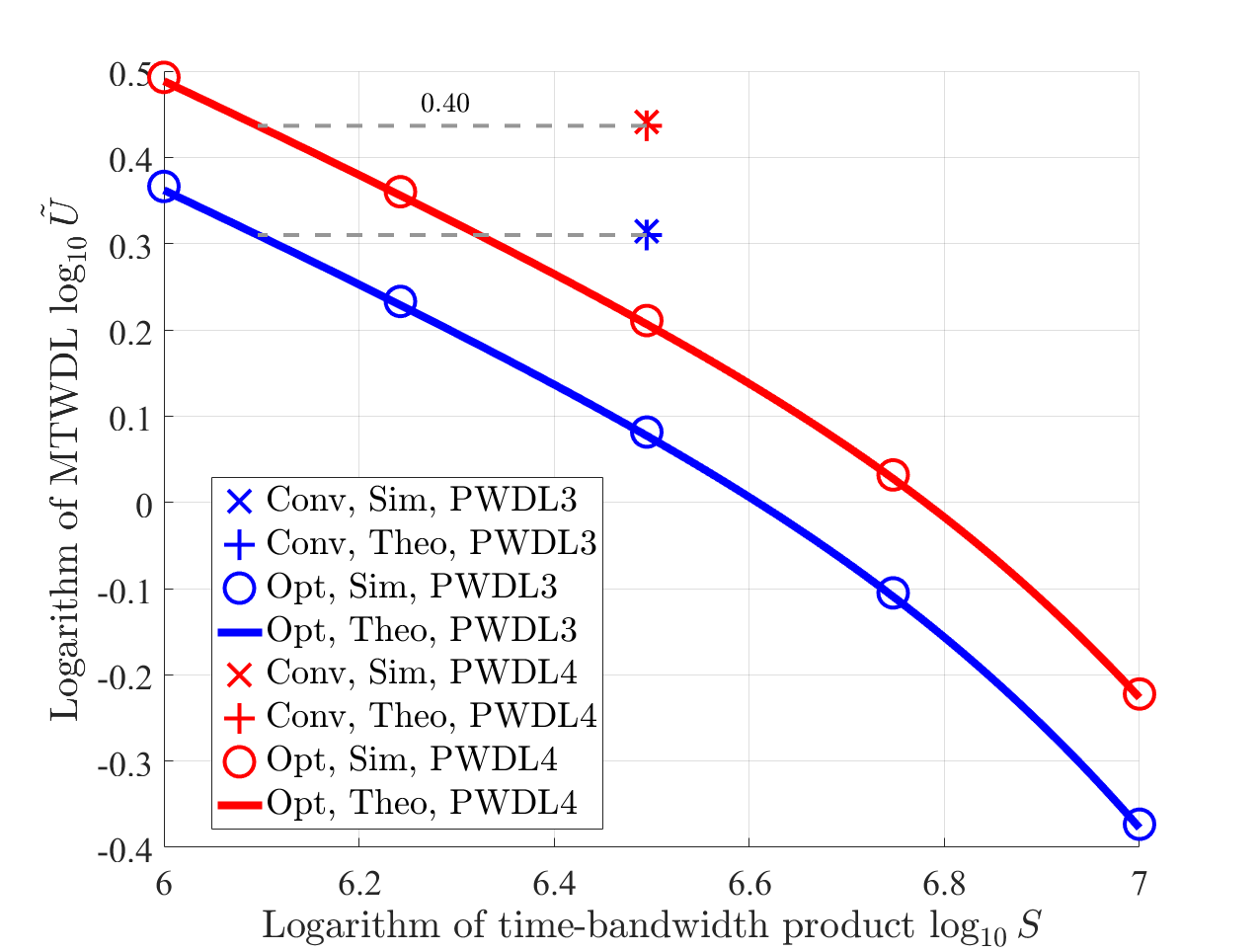}}
		\end{minipage}}
		\caption{Performance metric $\tilde{U}$ as a function of \gls{TBP} $S$.} 
		\label{fig:loss_TBP}
		\vspace{-0.4cm}
	\end{figure}

	\vspace{-0.3cm}
	\subsection{Comparison on MTWDL versus SNR}
	\label{subsec:comp_SNR}
	To demonstrate that the optimized design can reduce radar power, we show \gls{MTWDL} $\tilde{U}$ as a function of \gls{SNR} in Fig.~\ref{fig:loss_SNR}. The fixed \gls{TBP} is $S_\mathrm{con}$. The corresponding conventional and optimized design are 
	given in Section \ref{subsec:setting}. As the \gls{SNR} $\gamma$ varies from 5 dB to 25 dB, the \gls{MTWDL} $\tilde{U}$ are obtained. Similarly to Section~\ref{subsec:comp_TBP}, both simulated and theoretical results are shown.
	
	First, Fig. \ref{fig:loss_SNR} shows that the theoretical analysis fits the simulated results well when the \gls{SNR} is greater than 10 dB. Such an \gls{SNR} level is a common case in practice \cite{hasch2012millimeter}. Moreover, the optimized design outperforms the conventional design in terms of \gls{MTWDL} $\tilde{U}$. On the other hand, with \gls{TBP} fixed, the transmit power of the radar can be reduced by $4.0$ dB to reach a certain system performance using the optimized parameters. This result can be also deduced from \eqref{eqn:ratio_SNR}.
	
	
	\begin{figure}[t]
		\subfigure[PWDL 1 and PWDL 2.]{
			\begin{minipage}[t]{1.0\linewidth}
				\centering
				\centerline{\includegraphics[width=8cm]{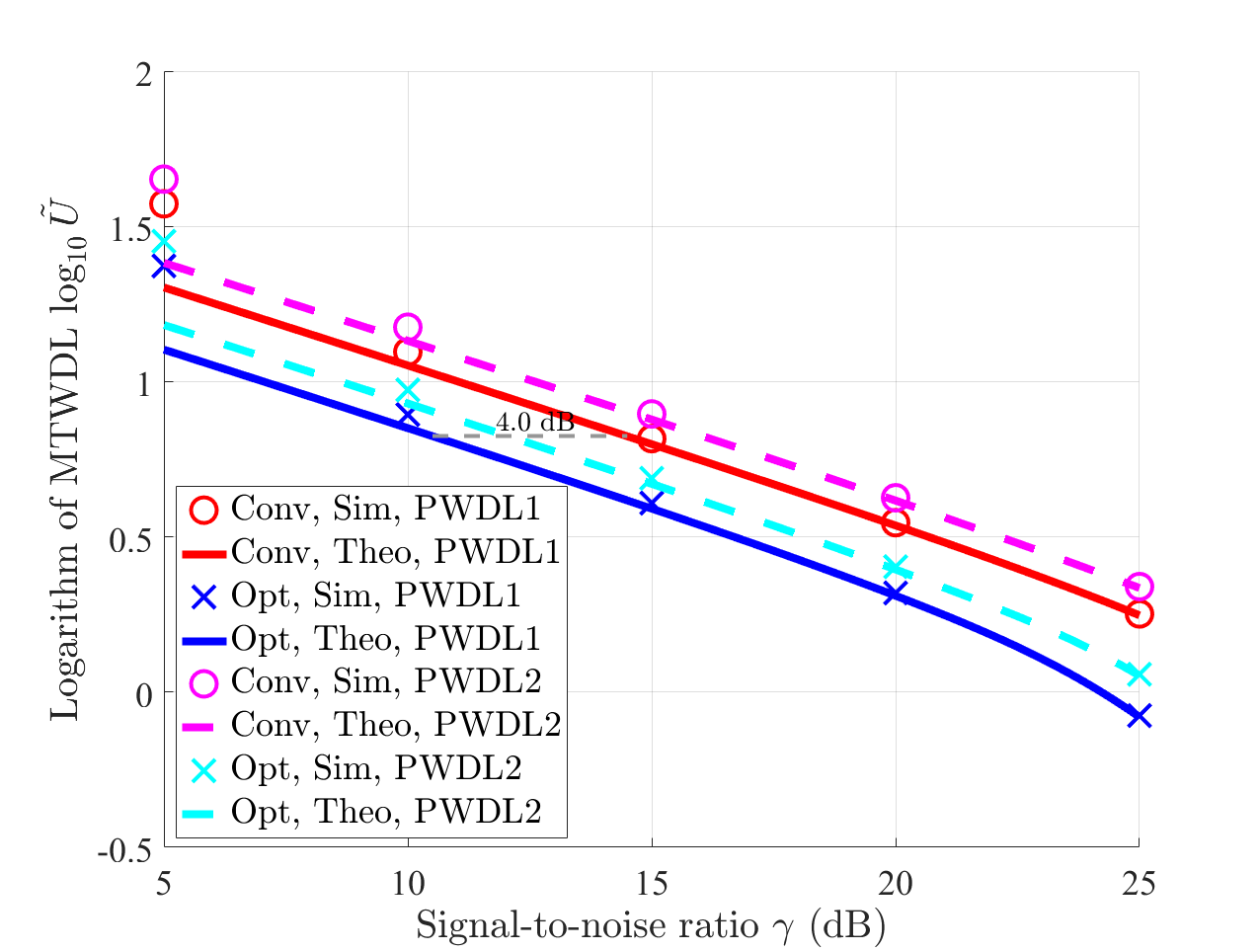}}
		\end{minipage}}
		\subfigure[PWDL 3 and PWDL 4.]{
			\begin{minipage}[t]{1.0\linewidth}
				\centering
				\centerline{\includegraphics[width=8cm]{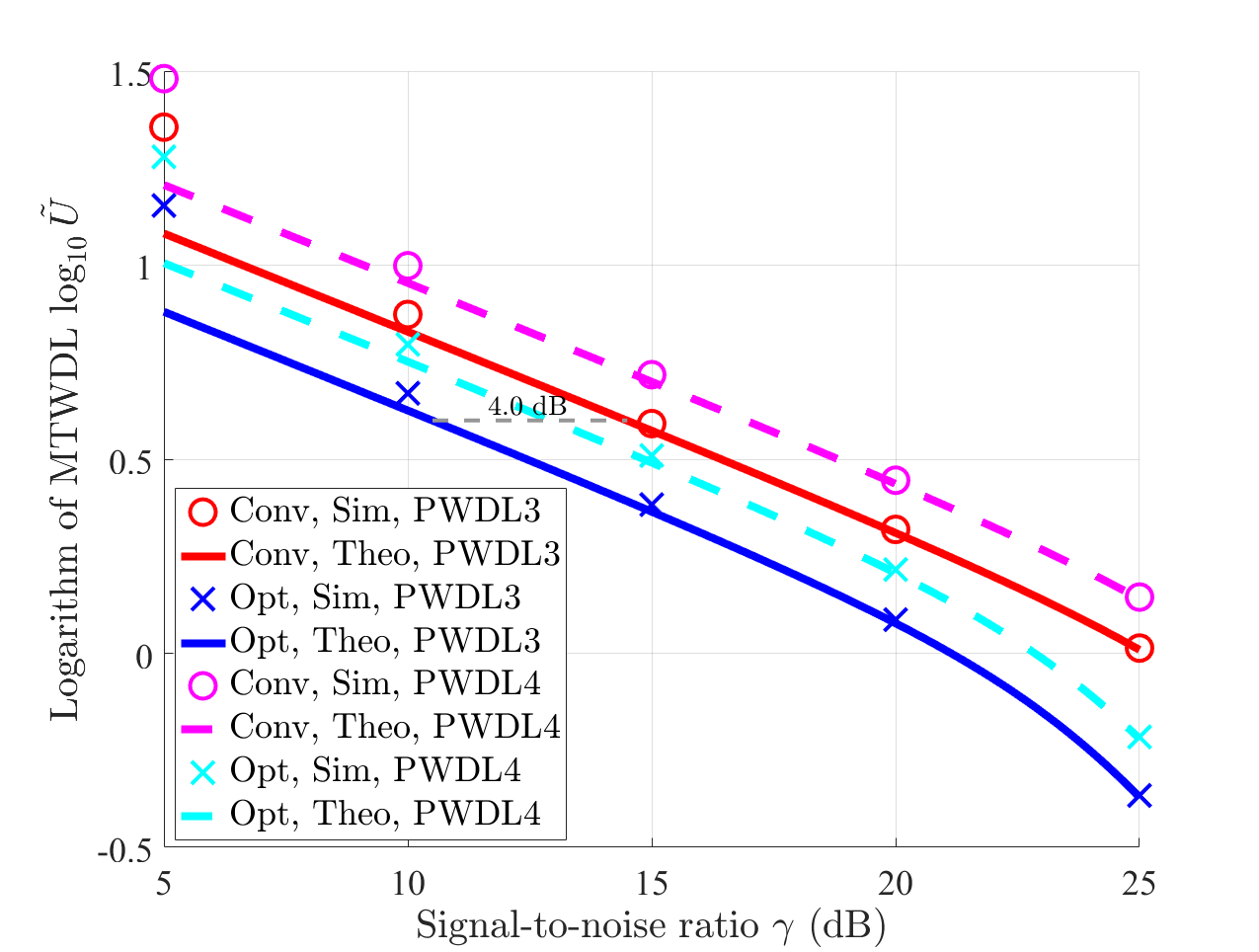}}
		\end{minipage}}
		\caption{Performance metric $\tilde{U}$ as a function of \gls{SNR} $\gamma$.} 
		\label{fig:loss_SNR}
		\vspace{-0.4cm}
	\end{figure}

	
	\section{Conclusion and outlook}
	\label{sec:conclusion}
	Facing the increasing tension of electromagnetic resources on the road, we propose a method of radar parameter design oriented to the performance of \gls{CWS}, which enables the radar to use the limited resources more efficiently. For this purpose, we propose the \gls{MTWDL} as the performance metric of CWS and establish its relationship with the performance of radar. Our analysis shows that in order  to achieve the  optimal performance of the \gls{CWS}, we  design the waveform bandwidth and duration to reduce the defined error index. Accordingly, we formulate the optimization on bandwidth and duration to minimize the error index with limited resources, which is quantified by \gls{TBP}. Both theoretical and numerical results show that the proposed method improves resource or energy efficiency compared with conventional parameter designs. 
	
	
	Establishment of the relationship between the performance of intelligent automotive systems and automotive radar waveform parameters is an important issue. It provides a deeper understanding of how many resources are needed and how to make full use of them. Future studies may involve more kinds of automotive systems, e.g., ACC systems, and give a further insight into the waveform design of automotive radar.
	
	\appendices
	\section{CLRB for estimating range and velocity}
	\label{appendix:CRLB}
	In order to obtain the \gls{CRLB} for estimating $\bm{\theta}_k$ from observations \eqref{eqn:y_nm2}, we first derive the \gls{CRLB} for $\bm{\xi}$, because $\bm{\theta}_k$ is a linear transform of $\bm{\xi}$. 
	The \gls{CRLB} for $\bm{\theta}_k$ is then obtained based on the results of $\bm{\xi}$. Making some approximations on this \gls{CRLB} under the assumption that $M$ is sufficiently large, yields the result in Proposition \ref{prop:crlb}.
	
	When we calculate the \gls{CRLB} for $\bm{\xi}$, our analysis is as follows: First, we calculate the \gls{FIM} for $\bm{\xi}$, which is the inverse of the \gls{CRLB}. Second, we prove that the \gls{FIM} can be approximated as a block diagonal matrix with large $N$ and $M$, whic facilitates the subsequent inverse operation on the \gls{FIM} and implies the \gls{CRLB} for $\bm{\xi}$. 
	

	
	Estimating $\bm{\xi}$ from $N\times M$ observations, as shown in \eqref{eqn:y_nm2}, is a problem of estimating two-dimensional frequencies \cite{hua1992estimating}. As defined in Section \ref{subsec:signal_model}, the parameter to be estimated $\bm{\xi}$ is written as $\boldsymbol{\xi}=[\boldsymbol{\xi}_1^T,\boldsymbol{\xi}_2^T,\dots,\boldsymbol{\xi}_K^T]^T$, where $\boldsymbol{\xi}_k=[b_k,\psi_k,f_{1,k},f_{2,k}]^T$. Here, variables $b_k$ and $\psi_k$ denote the amplitude and initial phase of the $k$-th two-dimensional frequency component, while $f_{1,k}$ and $f_{2,k}$ are its two-dimensional frequencies.
	The \gls{FIM} for $\bm{\xi}$, denoted as $\mathbf{F}$, is a $4K\times 4K$ matrix, which can be divided into $K^2$ $4\times 4$ submatrices, i.e.,
	\begin{equation} \label{eqn:fim_submatrix}
	\mathbf{F} = 
	\begin{bmatrix}
	\mathbf{F}_{11}  &  \cdots\ &\mathbf{F}_{1K}\\
	\vdots   & \ddots  & \vdots  \\
	\mathbf{F}_{K1} & \cdots\ &\mathbf{F}_{KK}\\
	\end{bmatrix}.
	\end{equation}
	According to \cite{hua1992estimating}, the submatrix $\mathbf{F}_{k_1 k_2}$ is given by \eqref{eqn:fim}. 
	\newcounter{mytempeqncnt}
	\begin{figure*}[!t]	
		\normalsize	
		\begin{equation} \label{eqn:fim}
		\mathbf{F}_{k_1 k_2} = \frac{2}{\sigma_w^2}
		\renewcommand\arraystretch{1.5}
		\begin{bmatrix}
		X^{(0,0,0)}_{k_1 k_2} & b_{k_2}X^{(1,0,0)}_{k_1 k_2} & b_{k_2}X^{(1,1,0)}_{k_1 k_2} & b_{k_2}X^{(1,0,1)}_{k_1 k_2}\\
		
		-b_{k_1}X^{(1,0,0)}_{k_1 k_2} & b_{k_1}b_{k_2}X^{(0,0,0)}_{k_1 k_2} &  b_{k_1}b_{k_2}X^{(0,1,0)}_{k_1 k_2} & b_{k_1}b_{k_2}X^{(0,0,1)}_{k_1 k_2}\\
		
		-b_{k_1}X^{(1,1,0)}_{k_1 k_2} & b_{k_1}b_{k_2}X^{(0,1,0)}_{k_1 k_2} &  b_{k_1}b_{k_2}X^{(0,2,0)}_{k_1 k_2} & b_{k_1}b_{k_2}X^{(0,1,1)}_{k_1 k_2}\\
		
		-b_{k_1}X^{(1,0,1)}_{k_1 k_2} & b_{k_1}b_{k_2}X^{(0,0,1)}_{k_1 k_2} &  b_{k_1}b_{k_2}X^{(0,1,1)}_{k_1 k_2} & b_{k_1}b_{k_2}X^{(0,0,2)}_{k_1 k_2}\\
		\end{bmatrix}.
		\end{equation}	
		\hrulefill	
		\vspace*{4pt}	
		\vspace{-0.3cm}
	\end{figure*}
	Here, $X^{(l_0,l_1,l_2)}_{k_1 k_2} (k_1=1,2,\dots,K, k_2=1,2,\dots,K)$ is given by
	\vspace{-0.1cm} 
	\begin{equation}
	\begin{aligned}
	X^{(l_0,l_1,l_2)}_{k_1 k_2} =& \sum_{n=0}^{N-1} \sum_{m=0}^{M-1}(2\pi n)^{l_1}(2\pi m)^{l_2}\\&\dot g_{l_0}(\Delta \psi
	+2\pi \Delta f_1 n+2\pi \Delta f_2 m),
	\end{aligned}
	\vspace{-0.1cm}
	\end{equation}
	where we define $g_{l_0}(x) = \cos(x)$ when $l_0=0$ and $g_{l_0}(x) = \sin(x)$ when $l_0=1$. The three new variables are defined as $\Delta \psi = \psi_{k_1} - \psi_{k_2}$, $\Delta f_1 = f_{1,k_1} - f_{1,k_2}$ and $\Delta f_2 = f_{2,k_1} - f_{2,k_2}$.
	Here, we drop the their dependence of on $k_1$ and $k_2$ for simplicity of notation.
	
	To study the asymptotic property of $\mathbf{F}$, we simplify the expression of $X^{(l_0,l_1,l_2)}_{k_1 k_2}$. First, when $l_1 = l_2 = 0$, we have
	\vspace{-0.1cm}
	\begin{equation} \label{eqn:entry_0}
	\begin{aligned}
	X^{(l_0,0,0)}_{k_1 k_2} \!=& \frac{\sin(\pi N \Delta f_1)}{\sin(\pi \Delta f_1)}
	\frac{\sin(\pi M \Delta f_2)}{\sin(\pi \Delta f_2)}\\
	&\cdot \! g_{l_0}(\Delta \psi
	\!+\! 2\pi\!(N\!-\!1) \Delta f_1 \!+\!2\pi (M\!-\!1) \Delta f_2).
	\end{aligned}
	\vspace{-0.1cm}
	\end{equation}
	When $l_1\neq 0$ or $l_2\neq 0$, 
	$X^{(l_0,l_1,l_2)}_{k_1 k_2} $ is the partial differential of  $X^{(l_0,0,0)}_{k_1 k_2}$ or $X^{(1,0,0)}_{k_1 k_2}$ over $\Delta f_1$ and $\Delta f_2$, given by
	\vspace{-0.1cm}
	\begin{equation} \label{eqn:entry_none_0}
	X^{(l_0,l_1,l_2)}_{k_1 k_2} = (-1)^{l_3} \frac{\partial^{l_1+l_2} X^{(l_0^{(0)},0,0)}}{(\partial \Delta f_1)^{l_1}(\partial \Delta f_2)^{l_2}},
	\vspace{-0.1cm}
	\end{equation}
	where $l_0^{(0)}=l_0+l_1+l_2 \mod 2$ ($p \mod q$ denotes the remainder after division of $p$ by $q$) and $l_3 = \lfloor (l_1+l_2+1-l_0^{(0)})/2 \rfloor$ ($\lfloor x \rfloor$ denotes the nearest integer less than or equal to $x$).
	
	To summarize the existing results, \eqref{eqn:fim_submatrix}, \eqref{eqn:fim}, \eqref{eqn:entry_0} and \eqref{eqn:entry_none_0} jointly give the \gls{FIM} for $\bm{\xi}$, i.e., $\mathbf{F}$. Specifically, \eqref{eqn:fim_submatrix} partitions $\mathbf{F}$ into submatrices; \eqref{eqn:fim} shows the entries of each submatrix; \eqref{eqn:entry_0} and \eqref{eqn:entry_none_0} give the simplified results of these entries.
	
	The \gls{CRLB} is the inverse of $\mathbf{F}$ \cite{kay1993fundamentals}. However, it is difficult to obtain an analytical expression for the inverse. In the following, we show that $\mathbf{F}$ can be approximated as a block diagonal matrix when 
	$N$ and $M$ approach infinity, which enables us to solve the inverse analytically. 
	
	We analyze the entries of $\mathbf{F}_{k_1 k_2}$ shown in \eqref{eqn:fim} when $N$ and $M$ tend to infinity. As a result, we have
	\vspace{-0.1cm}
	\begin{equation} \label{eqn:lim_X}
	\lim\limits_{\genfrac..{0pt}{2}{N \rightarrow \infty}{M \rightarrow \infty}}
	\frac{X^{(l_0,l_1,l_2)}_{k_1 k_2}}{N^{l_1+1}M^{l_2+1}}=
	\frac{(2\pi)^{l_1+l_2}}{(l_1+1)(l_2+1)} \delta_{l_0} \delta_{k_1-k_2},
	\vspace{-0.1cm}
	\end{equation}
	where $\delta_k$ is the discrete delta function, i.e., $\delta_k=1$ when $k=0$ and $\delta_k=0$ for the other values of $k$.
	That is, the limits are non-zero only if $k_1$ equals $k_2$, implying that the non-diagonal submatrices of $\mathbf{F}$ are negligible relative to the diagonal submatrices. Further, we introduce the matrix:
	\vspace{-0.1cm}
	\begin{equation} \label{eqn:Lambda}
	\mathbf{\Lambda}_k=\mbox{diag}\left(\sqrt{NM},b_k\sqrt{NM},b_k\sqrt{N^3M},b_k\sqrt{NM^3}\right).
	\vspace{-0.1cm}
	\end{equation}
	Based on the limits shown in \eqref{eqn:lim_X}, we can get
	\vspace{-0.1cm}
	\begin{equation} \label{eqn:fim_appr}
	\lim\limits_{\genfrac..{0pt}{2}{N \rightarrow \infty}{M \rightarrow \infty}}
	\mathbf{\Lambda}_{k_1}^{-1} \mathbf{F}_{k_1 k_2} \mathbf{\Lambda}_{k_2}^{-1}=
	{2\sigma_w^{-2}\delta_{k_1-k_2}}\mathbf{C},
	\vspace{-0.1cm}
	\end{equation}
	where the constant matrix $\mathbf{C}$ is 
	\vspace{-0.1cm}
	\begin{equation} \label{eqn:C}
	\mathbf{C}=
	\begin{bmatrix}
	1&0&0&0\\
	0&1&\pi&\pi\\
	0&\pi&{4\pi^2}/{3}&\pi^2\\
	0&\pi&\pi^2&{4\pi^2}/{3}\\
	\end{bmatrix},
	\vspace{-0.1cm}
	\end{equation}
	According to \eqref{eqn:fim_appr}, when $N$ and $M$ are sufficiently large, we can approximate the submatrix $\mathbf{F}_{k_1 k_2}$ as $\mathbf{F}^{(\infty)}_{k_1 k_2}$, given by
	\vspace{-0.1cm}
	\begin{equation} \label{eqn:F_infty}
	\mathbf{F}^{(\infty)}_{k_1 k_2} = 
	{2{\sigma_w^{-2}}\delta_{k_1-k_2}}\mathbf{\Lambda}_{k_1}\mathbf{C}\mathbf{\Lambda}_{k_2}.
	\vspace{-0.1cm}
	\end{equation}
	The delta function in \eqref{eqn:F_infty} indicates that the \gls{FIM} $\mathbf{F}$ can be approximated as a diagonal matrix as follows:
	\vspace{-0.1cm}
	\begin{equation} \label{eqn:F_inf}
	\begin{aligned}
	\mathbf{F}^{(\infty)} & =  \diag \left(\mathbf{F}^{(\infty)}_{11}, \mathbf{F}^{(\infty)}_{22}, \dots, \mathbf{F}^{(\infty)}_{KK}\right)\\
	& = 
	{2}{\sigma_w^{-2}}\diag(\mathbf{\Lambda}_{1}\mathbf{C}\mathbf{\Lambda}_{1},
	\mathbf{\Lambda}_{2}\mathbf{C}\mathbf{\Lambda}_{2},\dots,
	\mathbf{\Lambda}_{K}\mathbf{C}\mathbf{\Lambda}_{K}).
	\end{aligned}
	\vspace{-0.1cm}
	\end{equation} 
	Combining \eqref{eqn:Lambda} and \eqref{eqn:C}, the submatrix $\mathbf{\Lambda}_{k}\mathbf{C}\mathbf{\Lambda}_{k}$ is given by
	\begin{equation*}
	\mathbf{\Lambda}_{k}\mathbf{C}\mathbf{\Lambda}_{k} = 
	\begin{bmatrix}
	NM&0&0&0\\
	0&b_k^2 NM&\pi b_k^2 N^2M &\pi b_k^2 NM^2\\
	0&\pi b_k^2 N^2M&\frac{4\pi^2}{3}b_k^2 N^3M&\pi^2b_k^2 N^2M^2\\
	0&\pi b_k^2 NM^2&\pi^2b_k^2 N^2M^2&\frac{4\pi^2}{3}b_k^2 NM^3\\
	\end{bmatrix}.
	\end{equation*}
	
	We have approximated the \gls{FIM} as a block diagonal matrix $\mathbf{F}^{(\infty)}$, as shown in \eqref{eqn:F_inf}. The asymptotic \gls{CRLB} for $\bm{\xi}$ is the inverse of $\mathbf{F}^{(\infty)}$ and it is also a diagonal matrix, given by
	
	\vspace{-0.2cm}
	\begin{equation} \label{eqn:crlb_xi}
	\mathbf{B}_{\xi} = \diag\left(\mathbf{B}_{\xi,1},\mathbf{B}_{\xi,2},\dots,  \mathbf{B}_{\xi,K}\right),
	\vspace{-0.1cm}
	\end{equation}
	where $\mathbf{B}_{\xi,k}$ is the inverse of the counterpart in \eqref{eqn:F_inf}, given by
	\vspace{-0.1cm}
	\begin{equation*}
	\label{eqn:crlb_xi_k}
	\begin{aligned}
	\mathbf{B}_{\xi,k}&=\left({2}{\sigma_w^{-2}}\mathbf{\Lambda}_{k}\mathbf{C}\mathbf{\Lambda}_{k}\right)^{-1}\\
	&=\frac{\sigma_w^2}{2 A^2 \abs{\alpha_k}^2 NM}
	\left[
	\renewcommand\arraystretch{1.5}
	\begin{matrix}
	A^2 \abs{\alpha_k}^2 & 0 & 0 & 0\\
	0 & 7 & -\frac{3}{\pi N} & -\frac{3}{\pi M}\\
	0 & -\frac{3}{\pi N} & \frac{3}{\pi^2 N^2} & 0 \\
	0 & -\frac{3}{\pi M} & 0 & \frac{3}{\pi^2 M^2}
	\end{matrix}\right].
	\end{aligned}
	\vspace{-0.1cm}
	\end{equation*}
	The block diagonal structure of the \gls{CRLB} $\mathbf{B}_{\xi}$  means that the estimation errors of $\bm{\xi}_k$ with different $k$ are asymptotically uncorrelated when $N$ and $M$ are sufficiently large. In other words, the \gls{CRLB} for parameter estimation of one scatterer can be considered not influenced by others. 
	
	As mentioned in Section \ref{subsec:signal_model}, the estimate of range and velocity of the $k$-th scatter, formulated as $\hat{\bm{\theta}}_k = [\hat{d}_k, \hat{v}_k]^T$, is obtained  using the estimate of two-dimensional frequencies, formulated as  $\hat{\bm{f}}_k=[\hat{f}_{1,k},\hat{f}_{2,k}]^T$. To be specific, 
	according to \eqref{eqn:fd_fv} and \eqref{eqn:freq_2d}, $\hat{\bm{\theta}}_k$ is a linear combination of $\hat{\bm{f}}_k$, given by,
	\vspace{-0.1cm}
	\begin{equation}
	\begin{aligned}
	\hat{\bm{\theta}}_k 
	=\frac{c}{2}\left[
	\renewcommand\arraystretch{1.5}
	\begin{matrix}
	N/W & -{1}/{W}\\
	0 & {1}/({f_0 T_0})
	\end{matrix}\right]
	\end{aligned}
	\hat{\bm{f}}_k
	\vspace{-0.1cm}
	\end{equation}
	Thus, the \gls{CRLB} for $\bm{\theta}_k$ is obtained based on the \gls{CRLB} for ${\bm{f}}_k=[{f}_{1,k},{f}_{2,k}]^T$, which is contained in $\mathbf{B}_{\xi}$. As a result, its asymptotic \gls{CRLB} is
	\begin{equation} \label{eqn:crlb_0_theta_k}
	\begin{aligned}
	\mathbf{B}^{(0)}_{\theta,k}=
	\frac{3c^2}{8\pi^2 \gamma_k}
	\left[
	\renewcommand\arraystretch{1.5}
	\begin{matrix}
	\left(1+\frac{1}{M^2}\right)\frac{1}{W^2} & -\frac{1}{f_0 T_0 W M^2} \\
	-\frac{1}{f_0 T_0 W M^2} & \frac{1}{f_0^2 T^2}
	\end{matrix}\right].
	\end{aligned}
	\end{equation}
	
	To simplify the expression of \gls{CRLB} in \eqref{eqn:crlb_0_theta_k}, we conduct some approximations on $\mathbf{B}^{(0)}_{\theta,k}$ for large $M$. First, the \gls{CRLB} of $d_k$ is $\frac{3c^2}{8\pi^2 \gamma_k}\left(1+\frac{1}{M^2}\right)\frac{1}{W^2}$, which can be approximated as $\frac{3c^2}{8\pi^2 \gamma_k}\frac{1}{W^2}$ with large $M$. Second, \Copy{appr_crlb}{the correlation coefficient of estimation errors of $d_k$ and $v_k$ is $\frac{1}{\sqrt{M^2+1}}$ and it approaches $0$ when $M$ is sufficiently large, which means that estimation errors of $d_k$ and $v_k$ are asymptotically uncorrelated. Therefore, the \gls{CRLB} for $\bm{\theta}_k$ can be approximated as the diagonal matrix shown in \eqref{eqn:crlb_theta_k}.
		The result in \eqref{eqn:crlb_theta_k} can also be obtained using the stop-and-hop assumption
		, which is reasonable in automotive applications since the velocities of the vehicle and scatterers are much lower than the speed of light \cite{aydogdu2019radar}.}
	
	\label{subsec:crlb_dv}
	
	\section{Derivations of the GLRT rule}
	\label{appendix:glrt}
	The first step of \gls{GLRT} is to calculate the likelihood ratio $\mathcal{L}_G$ in \gls{GLRT} based on the measurement $\hat{\bm{\theta}}$, defined as \cite{kay1998fundamentals}
	\vspace{-0.1cm}
	\begin{equation}\label{eqn:glrt}
	\mathcal{L}_G\left(\hat{\bm{\theta}}\right) = 
	\left.{\max_{\bm{\theta}_1} p\left(\hat{\bm{\theta}}; \bm{\theta}_1,\mathcal{H}_1\right)}\right/
	{\max_{\bm{\theta}_0} p\left(\hat{\bm{\theta}}; \bm{\theta}_0,\mathcal{H}_0\right)},
	\vspace{-0.1cm}
	\end{equation}
	\Copy{gl}{where $p\left(\hat{\bm{\theta}}; \bm{\theta}_i,\mathcal{H}_i\right)$ denotes the likelihood of $\hat{\bm{\theta}}$ given the true value $\bm{\theta}_i$ and hypothesis $\mathcal{H}_i$.}
	
	The second step is to make decisions by comparing the logarithm likelihood ratio (LLR) with a threshold $\lambda_G$, i.e.,
	\vspace{-0.1cm}
	\begin{equation} \label{eqn:glrt_decision}
	\ln \mathcal{L}_G\left(\hat{\bm{\theta}}\right)
	\mathop  \gtrless \limits_{\mathcal{H}_0}^{\mathcal{H}_1}
	\lambda_G.
	\vspace{-0.1cm}
	\end{equation}
	
	As analyzed in Section \ref{sec:model}, the measurement error $\bm{\varepsilon}$ is assumed to be zero-mean Gaussian with covariance $\bm{\Sigma}_{\bm{\varepsilon}} = \diag(\sigma^2_{d},\sigma^2_{v})$. Thus, the likelihood $p(\hat{\bm{\theta}}; \bm{\theta}_i,\mathcal{H}_i)$ is given by
	\vspace{-0.1cm}
	\begin{equation*}
	\begin{aligned}
	p\left(\hat{\bm{\theta}}; \bm{\theta}_i,\mathcal{H}_i\right) \!=\! \frac{1}{2\pi \sqrt{\mbox{det }\bm{\Sigma}_{\bm{\varepsilon}}}} 
	e^{ -\frac{1}{2}\left(\hat{\bm{\theta}}-\bm{\theta}_i\right)^T
		\bm{\Sigma}_{\bm{\varepsilon}} ^{-1} 
		\left(\hat{\bm{\theta}}-\bm{\theta}_i\right)}.
	\end{aligned}
	\vspace{-0.1cm}
	\end{equation*}
	The LLR is given by
	\vspace{-0.1cm}
	\begin{equation}\label{eqn:glrt_log}
	\begin{aligned}
	\ln \mathcal{L}_G\left(\hat{\bm{\theta}}\right)\! = \!
	-\frac{1}{2}
	\left(
	\min_{\bm{\theta}_1\in \mathcal{R}_1}\left[\left(\hat{\bm{\theta}}-\bm{\theta}_1\right)^T\bm{\Sigma}_{\bm{\varepsilon}} ^{-1}\left(\hat{\bm{\theta}}-\bm{\theta}_1\right)\right] \right.\\
	-\left.
	\min_{\bm{\theta}_0\in \mathcal{R}_0}\left[\left(\hat{\bm{\theta}}-\bm{\theta}_0\right)^T\bm{\Sigma}_{\bm{\varepsilon}} ^{-1}\left(\hat{\bm{\theta}}-\bm{\theta}_0\right)\right] 
	\right).
	\end{aligned}
	\vspace{-0.1cm}
	\end{equation}
	Here, $\mathcal{R}_i$ denotes the set of $\bm{\theta}$ under the hypothesis $\mathcal{H}_i$, i.e., $\mathcal{R}_1 = \left\{\left.[d,v]^T \right| \  d+\tau_0 v<0, d>0\right\}$ and $\mathcal{R}_0 = \left\{\left.[d,v]^T \right| \  d+\tau_0 v\geq 0, d>0\right\}$.
	The right side of \eqref{eqn:glrt_log} involves two minimization problems, whose solutions are given by
	\vspace{-0.1cm}
	\begin{equation*} \label{eqn:s0}
	\begin{aligned}
	s_0\left(\hat{\bm{\theta}}\right) & = 
	\min_{\bm{\theta}_0\in \mathcal{R}_0}\left[\left(\hat{\bm{\theta}}-\bm{\theta}_0\right)^T\bm{\Sigma}_{\bm{\varepsilon}} ^{-1}\left(\hat{\bm{\theta}}-\bm{\theta}_0\right)\right]\\
	& = \left\{
	\begin{aligned}
	-\frac{\hat{d}+\tau_0 \hat{v}}{\sqrt{\sigma_{d}^2+\tau_0^2 \sigma_{v}^2}}, &\quad \hat{d}>0, \hat{d}<-\tau_0 \hat{v},\\
	0, &\quad \hat{d}>0, \hat{d}\geq -\tau_0 \hat{v}.
	\end{aligned}
	\right.
	\end{aligned}
	\end{equation*}
	\begin{equation*} \label{eqn:s1}
	\begin{aligned}
	s_1\left(\hat{\bm{\theta}}\right) & = 
	\min_{\bm{\theta}_1\in \mathcal{R}_1}\left[\left(\hat{\bm{\theta}}-\bm{\theta}_1\right)^T\bm{\Sigma}_{\bm{\varepsilon}} ^{-1}\left(\hat{\bm{\theta}}-\bm{\theta}_1\right)\right]\\
	& = \left\{
	\begin{aligned}
	0, &\quad \hat{d}>0, \hat{d}<-\tau_0 \hat{v}\\
	\frac{\hat{d}+\tau_0 \hat{v}}{\sqrt{\sigma_{d}^2+\tau_0^2 \sigma_{v}^2}}, &\quad \hat{d}\geq -\tau_0 \hat{v},\hat{d}\geq \frac{\sigma_{d}^2}{\sigma_{v}^2 \tau_0}\hat{v}, \\
	\sqrt{{\sigma_{d}^{-2}}{\hat{d}^2}+{\sigma_{v}^{-2}}{\hat{v}^2}}, & \quad \hat{d}>0,\hat{d}< \frac{\sigma_{d}^2}{\sigma_{v}^2 \tau_0}\hat{v}.
	\end{aligned}
	\right.
	\end{aligned}
	\vspace{-0.1cm}
	\end{equation*}
	Combining the results above, we get the expression of the LLR for \eqref{eqn:glrt_log}. Then, 
	we multiply both sides of \eqref{eqn:glrt_decision} by $-2\sqrt{\sigma_{d}^2+\tau_0^2 \sigma_{v}^2}$, yielding $T_G(\hat{\bm{\theta}}) \!=\! -2\sqrt{\sigma_{d}^2\!+\!\tau_0^2 \sigma_{v}^2} \ln \mathcal{L}_G(\hat{\bm{\theta}})$ and $\lambda \!=\! -2\sqrt{\sigma_{d}^2\!+\!\tau_0^2 \sigma_{v}^2}\  \lambda_G$.
	Thus, the \gls{GLRT} rule \eqref{eqn:glrt_decision} can be transformed into \eqref{eqn:glrt_rule}.
	
	\section{Monotonicity of MTWDL in terms of eror index}
	\label{appendix:mono_loss}
	We prove that $\tilde{U}(\sigma_{Z})$ is a monotonically increasing function of $\sigma_{Z}$ under the conditions stated in Theorem \ref{theo:mono}. The proof is divided into two steps: We first prove that with a fixed $\sigma_{Z}$, $\tilde{\lambda}(\sigma_{Z})$ is the only one extremum of $U(\lambda,\sigma_{Z})$. Then, we prove that the derivative ${\mathrm{d}\tilde{U}}/{\mathrm{d}\sigma_{Z}}$ is positive, indicating the increasing monotonicity of $\tilde{U}(\sigma_{Z})$. 
	
	We first divide the \gls{TWDL} shown in (\ref{eqn:twdl}) into two parts:
	\vspace{-0.1cm}
	\begin{equation*}
	\begin{aligned}
	U \!=\! \iint_{D_0} \! P_\mathrm{w}(d,v)  u(d,v) \mathrm{d}d \mathrm{d}v \!+\! \iint_{D_1} \! P_\mathrm{w}(d,v)  u(d,v) \mathrm{d}d \mathrm{d}v, 
	\end{aligned}
	\vspace{-0.1cm}
	\end{equation*}
	where the subdomains $D_0$ and $D_1$ are $D_0: (d,v)\in D, d+\tau_0 v\geq 0$ and $D_1: (d,v)\in D, d+\tau_0 v< 0$.
	According to (\ref{eqn:pw}) and the definition of Q-function  (\ref{eqn:qfunc}), we conclude that $P_w(d,v)$ is a smooth function with respect to $\lambda$ and $\sigma_{Z}$ for all $(d,v)$. Combined with the conditions that $u(d,v)$ is Lebesgue integrable over $D$ and $D$ is bounded, the \gls{TWDL} $U(\lambda,\sigma_{Z})$ is also a smooth function of $\lambda$ and $\sigma_{Z}$.

	As mentioned in Section \ref{subsec:metric}, for a fixed $\sigma_{Z}$, we aim to tune the threshold $\lambda$ to $\tilde{\lambda}(\sigma_Z)$ such that $U$ reaches its minimum. To this end, we treat $U(\lambda,\sigma_{Z})$ as a function of $\lambda$ and conduct partial differentiation $U(\lambda,\sigma_{Z})$ over $\lambda$ and get the first-order and second-order derivative. The results are shown in \eqref{eqn:diff_1} and \eqref{eqn:diff_2}. Here, we define $Z = d+\tau_0 v$. 
	
	\begin{figure*}[!htbp]	
		\normalsize	
		\begin{equation} \label{eqn:diff_1}
		\frac{\partial U}{\partial \lambda}\!=\!\frac{1}{\sqrt{2\pi}\sigma_{Z}}\!
		\left[\!\iint_{D_0}\!u(d,v) e^{-\frac{(Z\!-\!\lambda)^2}{2\sigma^2_{Z}}} \mathrm{d}d \mathrm{d}v 
		\!- \!
		\iint_{D_1} \!u(d,v) e^{-\frac{(Z\!-\!\lambda)^2}{2\sigma^2_{Z}}} \mathrm{d}d \mathrm{d}v\!\right].
		\end{equation}	
		\begin{equation} \label{eqn:diff_2}
		\begin{aligned}
		\frac{\partial^2 U}{\partial \lambda^2}=&\frac{1}{\sqrt{2\pi}\sigma^3_{Z}}
		\!\left[\!\iint_{D_0} u(d,v) e^{-\frac{(Z\!-\!\lambda)^2}{2\sigma^2_{Z}}}\!(Z\!-\!\lambda) \mathrm{d}d \mathrm{d}v 
		- 
		\iint_{D_1} u(d,v) e^{-\frac{(Z\!-\!\lambda)^2}{2\sigma^2_{Z}}}\!(Z\!-\!\lambda) \mathrm{d}d \mathrm{d}v\!\right]\\
		=&\frac{1}{\sqrt{2\pi}\sigma^3_{Z}}
		\iint_{D} u(d,v) e^{-\frac{(Z\!-\!\lambda)^2}{2\sigma^2_{Z}}}\abs{Z} \mathrm{d}d \mathrm{d}v 
		- \frac{\lambda}{\sqrt{2\pi}\sigma^3_{Z}}
		\left[\iint_{D_0} u(d,v) e^{-\frac{(Z\!-\!\lambda)^2}{2\sigma^2_{Z}}} \mathrm{d}d \mathrm{d}v 
		- 
		\iint_{D_1} u(d,v) e^{-\frac{(Z\!-\!\lambda)^2}{2\sigma^2_{Z}}} \mathrm{d}d \mathrm{d}v\right]\\
		=&\frac{1}{\sqrt{2\pi}\sigma^3_{Z}}
		\iint_{D} u(d,v) e^{-\frac{(Z\!-\!\lambda)^2}{2\sigma^2_{Z}}}\abs{Z} \mathrm{d}d \mathrm{d}v 
		-  \frac{\lambda}{\sigma^2_{Z}}
		\frac{\partial U}{\partial \lambda}.
		\end{aligned}
		\end{equation}	
		\hrulefill	
		\vspace*{4pt}
		\vspace{-0.3cm}	
	\end{figure*}
	
	We aim to find the minimum of $U(\lambda,\sigma_{Z})$ with fixed $\sigma_{Z}$ and an extremum is a candidate for minimum. Thus, we analyze the extrema of $U(\lambda,\sigma_{Z})$ and consider any $\lambda_0(\sigma_{Z})$ such that
	\vspace{-0.1cm}
	\begin{equation} \label{eqn:diff_1st_0}
	\left.\frac{\partial U}{\partial \lambda}\right|_{\lambda = \lambda_0,\sigma_{Z}} = 0.
	\vspace{-0.1cm}
	\end{equation}
	Combining (\ref{eqn:diff_2}), (\ref{eqn:diff_1st_0}) and  $u(d,v)>0$, we have
	\vspace{-0.1cm}
	\begin{equation} \label{eqn:diff_2nd_pos}
	\left.\frac{\partial^2 U}{\partial \lambda^2}\right|_{\lambda = \lambda_0,\sigma_{Z}} > 0.
	\vspace{-0.1cm}
	\end{equation}
	This means that all extrema are local minima. With this result and the fact that $U(\lambda,\sigma_{Z})$ is smooth, we can further conclude that $U(\lambda,\sigma_{Z})$ has only one extremum, which is also the global minimum. Therefore, the optimal threshold $\tilde{\lambda}(\sigma_{Z})$ satisfies
	\vspace{-0.1cm}
	\begin{equation} \label{eqn:diff_1st_opt}
	\left.\frac{\partial U}{\partial \lambda}\right|_{\lambda = \tilde{\lambda}(\sigma_{Z}),\sigma_{Z}} = 0.
	\vspace{-0.1cm}
	\end{equation}
	
	Now we have proven that the optimal threshold $\tilde{\lambda}(\sigma_{Z})$ is the only one extremum of $U(\lambda, \sigma_Z)$ for a fixed $\sigma_{Z}$. After $\tilde{\lambda}(\sigma_{Z})$ is obtained, the minimal $U$ is a function of $\sigma_Z$, i.e., $\tilde{U}(\sigma_{Z}) = U(\tilde{\lambda}(\sigma_{Z}),\sigma_{Z})$. Our main concern is the monotonicity of  $\tilde{U}(\sigma_{Z}) $. With the chain rule, its derivative is as follows:
	\vspace{-0.3cm}
	\begin{equation} \label{eqn:diff_opt_loss}
	\frac{\mathrm{d}\tilde{U}}{\mathrm{d}\sigma_{Z}}
	= \left.\frac{\partial U}{\partial \lambda}\right|_{\lambda = \tilde{\lambda}}\cdot 
	\frac{\partial \tilde{\lambda}}{\partial \sigma_{Z}}
	+ \left.\frac{\partial U}{\partial\sigma_{Z}}\right|_{\lambda = \tilde{\lambda}} = 
	\left.\frac{\partial U}{\partial\sigma_{Z}}\right|_{\lambda = \tilde{\lambda}}.
	\vspace{-0.1cm}
	\end{equation}
	Note that the above deduction uses the result \eqref{eqn:diff_1st_opt}. The partial differential $\partial U/\partial \sigma_{Z}$ can be expanded as follows:
	\vspace{-0.1cm}
	\begin{equation} \label{eqn:diff_U_sigma}
	\begin{aligned}
	\frac{\partial U}{\partial \sigma_{Z}}=&\frac{1}{\sqrt{2\pi}\sigma^2_{Z}}
	\left[\iint_{D_0} u(d,v) e^{-\frac{(Z\!-\!\lambda)^2}{2\sigma^2_{Z}}}\!(Z\!-\!\lambda) \mathrm{d}d \mathrm{d}v 
	\right. \\
	&-\left.
	\iint_{D_1} u(d,v) e^{-\frac{(Z\!-\!\lambda)^2}{2\sigma^2_{Z}}}\!(Z\!-\!\lambda) \mathrm{d}d \mathrm{d}v\right].
	\end{aligned}
	\vspace{-0.1cm}
	\end{equation}
	Comparing \eqref{eqn:diff_U_sigma} with the first line of \eqref{eqn:diff_2}, we have
	\vspace{-0.1cm}
	\begin{equation} \label{eqn:diff_opt_loss2}
	\frac{\partial U}{\partial\sigma_{Z}} = \sigma_Z
	\frac{\partial^2 U}{\partial \lambda^2}.
	\vspace{-0.1cm}
	\end{equation}
	Combining (\ref{eqn:diff_2nd_pos}), (\ref{eqn:diff_opt_loss}) and (\ref{eqn:diff_opt_loss2}), we have
	\vspace{-0.1cm}
	\begin{equation}
	\frac{\mathrm{d}\tilde{U}}{\mathrm{d}\sigma_{Z}}
	=\left.\frac{\partial U}{\partial\sigma_{Z}}\right|_{\lambda = \tilde{\lambda}}
	= \left.\sigma_Z
	\frac{\partial^2 U}{\partial \lambda^2}\right|_{\lambda = \tilde{\lambda}}
	> 0.
	\vspace{-0.1cm}
	\end{equation}
	Therefore, $\tilde{U}(\sigma_{Z})$ is an increasing function of $\sigma_{Z}$.

	\section*{Acknowledgment}
	The authors would like to thank Prof. Huadong Meng with California PATH, University of California, Berkeley, for his fruitful discussions and his supports in our early work.
	
	\bibliographystyle{IEEEtran}
	\bibliography{IEEEabrv,2D}

\begin{thebibliography}{10}
\providecommand{\url}[1]{#1}
\csname url@samestyle\endcsname
\providecommand{\newblock}{\relax}
\providecommand{\bibinfo}[2]{#2}
\providecommand{\BIBentrySTDinterwordspacing}{\spaceskip=0pt\relax}
\providecommand{\BIBentryALTinterwordstretchfactor}{4}
\providecommand{\BIBentryALTinterwordspacing}{\spaceskip=\fontdimen2\font plus
\BIBentryALTinterwordstretchfactor\fontdimen3\font minus
  \fontdimen4\font\relax}
\providecommand{\BIBforeignlanguage}[2]{{%
\expandafter\ifx\csname l@#1\endcsname\relax
\typeout{** WARNING: IEEEtran.bst: No hyphenation pattern has been}%
\typeout{** loaded for the language `#1'. Using the pattern for}%
\typeout{** the default language instead.}%
\else
\language=\csname l@#1\endcsname
\fi
#2}}
\providecommand{\BIBdecl}{\relax}
\BIBdecl

\bibitem{shaout2011advanced}
A.~Shaout, D.~Colella, and S.~Awad, ``Advanced driver assistance systems-past,
  present and future,'' in \emph{Proc. 7th Int. Comput. Eng. Conf.}\hskip 1em
  plus 0.5em minus 0.4em\relax IEEE, 2011, pp. 72--82.

\bibitem{vahidi2003research}
A.~Vahidi and A.~Eskandarian, ``Research advances in intelligent collision
  avoidance and adaptive cruise control,'' \emph{IEEE Trans. Intell. Transp.
  Syst.}, vol.~4, no.~3, pp. 143--153, 2003.

\bibitem{mukhtar2015vehicle}
A.~Mukhtar, L.~Xia, and T.~B. Tang, ``Vehicle detection techniques for
  collision avoidance systems: A review,'' \emph{IEEE Trans. Intell. Transp.
  Syst.}, vol.~16, no.~5, pp. 2318--2338, 2015.

\bibitem{lee2016real}
D.~Lee and H.~Yeo, ``Real-time rear-end collision-warning system using a
  multilayer perceptron neural network,'' \emph{IEEE Trans. Intell. Transp.
  Syst.}, vol.~17, no.~11, pp. 3087--3097, 2016.

\bibitem{li2011model}
S.~Li, K.~Li, R.~Rajamani, and J.~Wang, ``Model predictive multi-objective
  vehicular adaptive cruise control,'' \emph{IEEE Trans. Control Syst.
  Technol.}, vol.~19, no.~3, pp. 556--566, 2011.

\bibitem{skolnik1962introduction}
M.~I. Skolnik, ``Introduction to radar,'' \emph{Radar handbook}, vol.~2, p.~21,
  1962.

\bibitem{rohling2001waveform}
H.~Rohling and M.-M. Meinecke, ``Waveform design principles for automotive
  radar systems,'' in \emph{Proc. IEEE 2001 CIE Int. Conf. Radar}, vol.~4,
  2001, p.~1.

\bibitem{kunert2012eu}
M.~Kunert, ``The {EU} project {MOSARIM}: A general overview of project
  objectives and conducted work,'' in \emph{Proc. 9th Eur. Radar Conf.
  (EuRAD)}.\hskip 1em plus 0.5em minus 0.4em\relax IEEE, 2012, pp. 1--5.

\bibitem{Brooker2007Mutual}
G.~M. Brooker, ``Mutual interference of millimeter-wave radar systems,''
  \emph{IEEE Trans. Electromagn. Compat.}, vol.~49, no.~1, pp. 170--181, 2007.

\bibitem{goppelt2010automotive}
M.~Goppelt, H.~L. Blocher, and W.~Menzel, ``Automotive radar – investigation
  of mutual interference mechanisms,'' \emph{Adv. Radio Sci.}, vol.~8, pp.
  55--60, 2010.

\bibitem{braun2013co-channel}
M.~Braun, R.~Tanbourgi, and F.~K. Jondral, ``Co-channel interference
  limitations of {OFDM} communication-radar networks,'' \emph{EURADSIP J.
  Wireless Commun. Netw.}, vol. 2013, no.~1, p. 207, 2013.

\bibitem{alhourani2018stochastic}
A.~Alhourani, R.~J. Evans, S.~Kandeepan, B.~Moran, and H.~Eltom, ``Stochastic
  geometry methods for modeling automotive radar interference,'' \emph{IEEE
  Trans. Intell. Transp. Syst.}, vol.~19, no.~2, pp. 1--12, 2018.

\bibitem{fischer2011minimizing}
C.~Fischer, M.~Goppelt, H.-L. Bl{\"o}cher, and J.~Dickmann, ``Minimizing
  interference in automotive radar using digital beamforming,'' \emph{Adv.
  Radio Sci.}, vol.~9, no. B. 2, pp. 45--48, 2011.

\bibitem{sun2015interference}
S.~Sun, ``Interference mitigation in automotive radars,'' 2015.

\bibitem{papadimitratos2009vehicular}
P.~Papadimitratos, A.~De~La~Fortelle, K.~Evenssen, R.~Brignolo, and S.~Cosenza,
  ``Vehicular communication systems: Enabling technologies, applications, and
  future outlook on intelligent transportation,'' \emph{IEEE Commun. Mag.},
  vol.~47, no.~11, 2009.

\bibitem{ruan2016sharing}
H.~Ruan, Y.~Liu, H.~Meng, and X.~Wang, ``Sharing for safety: The bandwidth
  allocation among automotive radars,'' in \emph{Proc. IEEE Int. Conf. Global
  Signal Process. (GlobalSIP)}, 2016, pp. 1047--1051.

\bibitem{khoury2016radarmac}
J.~Khoury, R.~Ramanathan, D.~McCloskey, R.~Smith, and T.~Campbell,
  ``Radar{MAC}: Mitigating radar interference in self-driving cars,'' in
  \emph{Proc. IEEE Int. Conf. Sens., Commun. Netw.}, 2016, pp. 1--9.

\bibitem{aydogdu2019radar}
C.~Aydogdu, G.~K. Carvajal, O.~Eriksson, H.~Hellsten, H.~Herbertsson, M.~F.
  Keskin, E.~Nilsson, M.~Rydstr{\"o}m, K.~Van{\"a}s, and H.~Wymeersch, ``Radar
  interference mitigation for automated driving,'' \emph{arXiv preprint
  arXiv:1909.09441}, 2019.

\bibitem{aydogdu2019radarchat}
C.~{Aydogdu}, M.~F. {Keskin}, N.~{Garcia}, H.~{Wymeersch}, and D.~W. {Bliss},
  ``Radchat: Spectrum sharing for automotive radar interference mitigation,''
  \emph{IEEE Trans. Intell. Transp. Syst.}, pp. 1--14, 2019.

\bibitem{russell1997millimeter}
M.~E. Russell, A.~Crain, A.~Curran, R.~A. Campbell, C.~A. Drubin, and W.~F.
  Miccioli, ``Millimeter-wave radar sensor for automotive intelligent cruise
  control ({ICC}),'' \emph{IEEE Trans. Microw. Theory Techn.}, vol.~45, no.~12,
  pp. 2444--2453, 1997.

\bibitem{hasch2012millimeter}
J.~Hasch, E.~Topak, R.~Schnabel, T.~Zwick, R.~Weigel, and C.~Waldschmidt,
  ``Millimeter-wave technology for automotive radar sensors in the 77 {GH}z
  frequency band,'' \emph{IEEE Trans. Microw. Theory Techn.}, vol.~60, no.~3,
  pp. 845--860, 2012.

\bibitem{sit2018bpsk}
Y.~L. Sit, G.~Li, S.~Manchala, H.~Afrasiabi, C.~Sturm, and U.~Lubbert,
  ``{BPSK}-based {MIMO} {FMCW} automotive-radar concept for 3{D} position
  measurement,'' in \emph{Proc. 9th Eur. Radar Conf. (EuRAD)}.\hskip 1em plus
  0.5em minus 0.4em\relax IEEE, 2018, pp. 289--292.

\bibitem{yi201924}
X.~Yi, G.~Feng, Z.~Liang, C.~Wang, B.~Liu, C.~Li, K.~Yang, C.~C. Boon, and
  Q.~Xue, ``A 24/77 {GHz} dual-band receiver for automotive radar
  applications,'' \emph{IEEE Access}, 2019.

\bibitem{folster2005automotive}
F.~Folster, H.~Rohling, and U.~Lubbert, ``An automotive radar network based on
  77 {GH}z {FMCW} sensors,'' in \emph{IEEE Int. Radar Conf., 2005.}\hskip 1em
  plus 0.5em minus 0.4em\relax IEEE, 2005, pp. 871--876.

\bibitem{forstner200877ghz}
H.~Forstner, H.~Knapp, H.~Jager, E.~Kolmhofer, J.~Platz, F.~Starzer, M.~Treml,
  A.~Schinko, G.~Birschkus, J.~Bock \emph{et~al.}, ``A 77{GH}z 4-channel
  automotive radar transceiver in {S}i{G}e,'' in \emph{2008 IEEE Radio
  Frequency Integrated Circuits Symp.}\hskip 1em plus 0.5em minus 0.4em\relax
  IEEE, 2008, pp. 233--236.

\bibitem{braun2009parametrization}
M.~Braun \emph{et~al.}, ``Parametrization of joint {OFDM}-based radar and
  communication systems for vehicular applications,'' in \emph{20th Int. Symp.
  Pers., Indoor and Mobile Radio Commun.}\hskip 1em plus 0.5em minus
  0.4em\relax IEEE, 2009, pp. 3020--3024.

\bibitem{turlapaty2014range}
A.~Turlapaty, Y.~Jin, and Y.~Xu, ``Range and velocity estimation of radar
  targets by weighted {OFDM} modulation,'' in \emph{Proc. IEEE Radar Conf.},
  2014, pp. 1358--1362.

\bibitem{bell1993information}
M.~R. Bell, ``Information theory and radar waveform design,'' \emph{IEEE Trans.
  Inf. Theory}, vol.~39, no.~5, pp. 1578--1597, 1993.

\bibitem{leshem2007information}
A.~Leshem, O.~Naparstek, and A.~Nehorai, ``Information theoretic adaptive radar
  waveform design for multiple extended targets,'' \emph{IEEE J. Sel. Topics
  Signal Process.}, vol.~1, no.~1, pp. 42--55, 2007.

\bibitem{lenz2018joint}
A.~Lenz, M.~S. Stein, and A.~L. Swindlehurst, ``Joint transmit and receive
  filter optimization for sub-nyquist delay-doppler estimation,'' \emph{IEEE
  Trans. Signal Process}, vol.~66, no.~10, pp. 2542--2556, 2018.

\bibitem{knapp1976gneralized}
C.~{Knapp} and G.~{Carter}, ``The generalized correlation method for estimation
  of time delay,'' \emph{IEEE Trans. Acoust., Speech, Signal Process.},
  vol.~24, no.~4, pp. 320--327, 1976.

\bibitem{chen2008mimo}
C.~{Chen} and P.~P. {Vaidyanathan}, ``{MIMO} radar ambiguity properties and
  optimization using frequency-hopping waveforms,'' \emph{IEEE Trans. Signal
  Process}, vol.~56, no.~12, pp. 5926--5936, 2008.

\bibitem{stein2014information}
M.~{Stein}, M.~{Castañeda}, A.~{Mezghani}, and J.~A. {Nossek},
  ``Information-preserving transformations for signal parameter estimation,''
  \emph{IEEE Signal Process. Lett.}, vol.~21, no.~7, pp. 866--870, 2014.

\bibitem{HEIDENREICH20133400}
P.~Heidenreich and A.~M. Zoubir, ``Fast maximum likelihood {DOA} estimation in
  the two-target case with applications to automotive radar,'' \emph{Signal
  Process.}, vol.~93, no.~12, pp. 3400--3409, 2013.

\bibitem{kay1993fundamentals}
S.~M. Kay, \emph{Fundamentals of statistical signal processing}.\hskip 1em plus
  0.5em minus 0.4em\relax Prentice Hall PTR, 1993.

\bibitem{stove1992linear}
A.~G. Stove, ``Linear {FMCW} radar techniques,'' in \emph{IEE Proc. F (Radar
  and Signal Process.)}, vol. 139, no.~5.\hskip 1em plus 0.5em minus
  0.4em\relax IET, 1992, pp. 343--350.

\bibitem{hyun2010two}
E.~{Hyun}, W.~{Oh}, and J.~{Lee}, ``Two-step moving target detection algorithm
  for automotive 77 {GH}z {FMCW} radar,'' in \emph{2010 IEEE 72nd Veh. Technol.
  Conf - Fall}, 2010, pp. 1--5.

\bibitem{hua1992estimating}
Y.~Hua, ``Estimating two-dimensional frequencies by matrix enhancement and
  matrix pencil,'' \emph{IEEE Trans. Signal Process.}, vol.~40, no.~9, pp.
  2267--2280, 1992.

\bibitem{massey1951kolmogorov}
F.~J. Massey~Jr, ``The {K}olmogorov-{S}mirnov test for goodness of fit,''
  \emph{J. Amer. Statist. Assoc.}, vol.~46, no. 253, pp. 68--78, 1951.

\bibitem{lee2005evaluation}
K.~Lee and H.~Peng, ``Evaluation of automotive forward collision warning and
  collision avoidance algorithms,'' \emph{Veh. Syst. Dynamics}, vol.~43,
  no.~10, pp. 735--751, 2005.

\bibitem{minderhoud2001extended}
M.~M. Minderhoud and P.~H. Bovy, ``Extended time-to-collision measures for road
  traffic safety assessment,'' \emph{Accident Anal. \& Prevention}, vol.~33,
  no.~1, pp. 89--97, 2001.

\bibitem{jermakian2017effects}
J.~S. Jermakian, S.~Bao, M.~L. Buonarosa, J.~R. Sayer, and C.~M. Farmer,
  ``Effects of an integrated collision warning system on teenage driver
  behavior,'' \emph{J. of Safety Res.}, vol.~61, pp. 65--75, 2017.

\bibitem{phillips2019real}
D.~J. Phillips, J.~C. Aragon, A.~Roychowdhury, R.~Madigan, S.~Chintakindi, and
  M.~J. Kochenderfer, ``Real-time prediction of automotive collision risk from
  monocular video,'' \emph{arXiv preprint arXiv:1902.01293}, 2019.

\bibitem{kay1998fundamentals}
S.~M. Kay, ``Fundamentals of statistical signal processing, vol. ii: Detection
  theory,'' \emph{Signal Processing. Upper Saddle River, NJ: Prentice Hall},
  1998.

\bibitem{moon2009design}
S.~Moon, I.~Moon, and K.~Yi, ``Design, tuning, and evaluation of a full-range
  adaptive cruise control system with collision avoidance,'' \emph{Control Eng.
  Practice}, vol.~17, no.~4, pp. 442--455, 2009.

\end{thebibliography}
	
	
	
	
	
	

\end{document}